\documentclass[journal,twoside,web]{ieeecolor}
\usepackage{generic}
\usepackage{cite}
\usepackage{amsmath,amssymb,amsfonts}
\usepackage{algorithmic}
\usepackage{graphicx}
\usepackage{algorithm,algorithmic}
\usepackage{hyperref}
\usepackage{textcomp}
\def\BibTeX{{\rm B\kern-.05em{\sc i\kern-.025em b}\kern-.08em
    T\kern-.1667em\lower.7ex\hbox{E}\kern-.125emX}}
\usepackage{cuted}
\newtheorem{lemma}{Lemma}
\newtheorem{remark}{Remark}
\newtheorem{theorem}{Theorem}
\newtheorem{assumption}{Assumption}

\begin{document}
\title{Robust Macroscopic Density Control of Heterogeneous Multi-Agent Systems}
%
%
%
\author{Gian Carlo Maffettone, Davide Salzano, and Mario di Bernardo
\thanks{Gian Carlo Maffettone and Davide Salzano contributed equally to this work (Corresponding author: Mario di Bernardo).}
\thanks{Gian Carlo Maffettone is with the Modeling and Engineering Risk and Complexity program of the Scuola Superiore Meridionale, Naples, Italy (e-mails: giancarlo.maffettone@unina.it, gc.maffettone@ssmeridionale.it). }
\thanks{Davide Salzano and Mario di Bernardo are with the Department of Electrical Engineering and Information Technology of the University of Naples Federico II, Naples, Italy (e-mails: davide.salzano@unina.it,m mario.dibernardo@unina.it).}
\thanks{Mario di Bernardo is also with the Modeling and Engineering Risk and Complexity program of the Scuola Superiore Meridionale, Naples, Italy.}}
\maketitle
\begin{abstract}
Modern applications, such as orchestrating the collective behavior of robotic swarms or traffic flows, require the coordination of large groups of agents evolving in unstructured environments, where disturbances and unmodeled dynamics are unavoidable. In this work, we develop a scalable macroscopic density control framework in which a feedback law is designed directly at the level of an advection--diffusion partial differential equation. We formulate the control problem in the density space and prove global exponential convergence towards the desired behavior in $\mathcal{L}^2$ with guaranteed asymptotic rejection of bounded unknown drift terms, 
explicitly accounting for heterogeneous agent dynamics, unmodeled behaviors, 
and environmental perturbations. Our theoretical findings are corroborated by numerical experiments spanning heterogeneous oscillators, traffic systems, and swarm robotics in partially unknown environments.
\end{abstract}


\thispagestyle{empty}

\section{Introduction}
\label{sec:introduction}
\IEEEPARstart{L}{arge-scale} multi-agent systems are ubiquitous in modern society, with applications ranging from swarm coordination in robotics and biology \cite{dorigo2021swarm, massana2022rectification, palacci2013living, salzano2025vivo}, to traffic control \cite{siri2021freeway} and environmental management \cite{zahugi2013oil, pashna2020autonomous}. 
In such systems, macroscopic collective behaviors typically emerge from microscopic interactions among a large number of agents. 
Classic microscopic (i.e., agent-based) control approaches aim at orchestrating collective behavior by designing local interaction and communication protocols based on individual agent dynamics \cite{oh2015survey, cortes2009global}. 
For instance, in \cite{della2020stochastic}, synchronization of a large stochastic population is achieved by controlling a small fraction of agents, while in \cite{jiang2018fully} an ensemble of heterogeneous systems is guided toward a prescribed formation via fully distributed controllers.
These works demonstrate that distributed control strategies can effectively induce collective behaviors even in the presence of stochasticity and heterogeneity. 
However, microscopic controllers are typically limited to relatively simple macroscopic objectives, such as geometric formations \cite{giusti2023distributed}, synchronization \cite{yu2009pinning}, or flocking \cite{beaver2021overview}. Moreover, providing analytical guarantees of stability and robustness becomes increasingly challenging as the population size grows \cite{d2023controlling}.

A recent paradigmatic shift in multi-agent control leverages macroscopic descriptions of collective dynamics. Within this perspective, which forms the basis of multi-scale control methodologies for large agent populations as recently proposed in \cite{d2023controlling,maffettone2022continuification}, collective behavior is captured through the evolution of macroscopic observables \cite{elamvazhuthi2019mean, lama2025nonreciprocal}. For instance, the spatial organization of mobile agents can be more compactly represented by a density function than by tracking individual states. 
Starting from agent-level dynamics and interactions, continuum models based on partial differential equations (PDEs) are derived to describe the evolution of the population density. Control laws can then be designed at the macroscopic level and systematically bridged to microscopic actuation through discretization and spatial sampling. This macro-to-micro architecture, often referred to as continuification control \cite{nikitin2021continuation, maffettone2022continuification}, draws upon tools from PDE control \cite{krstic2008boundary} and mean-field control \cite{fornasier2014mean, ascione2023mean}, and has enabled the regulation of complex collective behaviors in very large populations \cite{maffettone2024mixed, nikitin2021boundary, albi2021optimized, lin2025heterogeneous, boldini2024stigmergy}.
In particular, density control frameworks \cite{chen2023density} allow the specification of rich macroscopic objectives that go well beyond classical tasks such as synchronization, formation control, and flocking. When embedded within a multi-scale control architecture, such macroscopic designs provide scalable strategies with analytical guarantees that can be preserved under microscopic implementation.

Despite these advantages, most existing macroscopic density control methods rely on simplified agent models and typically neglect heterogeneity, unmodeled dynamics, and external disturbances. In particular, most available approaches assume known and homogeneous agent dynamics and provide only limited robustness characterizations, often in the form of bounded steady-state errors, rather than explicit feedback designs guaranteeing asymptotic regulation under uncertainty. This limitation becomes particularly critical in multi-scale macro-to-micro control architectures, where macroscopic guarantees must remain valid after discretization and microscopic actuation.

Motivated by these limitations, in this paper we develop a continuification-based control architecture for the robust regulation of multi-agent systems affected by bounded unmodeled dynamics and disturbances. Specifically, we formulate a density tracking problem for large populations of stochastic and heterogeneous agents and, by exploiting a continuum PDE description of the population dynamics, we construct macroscopic upper and lower bounding systems that account for unknown drift terms. A unified macroscopic feedback control law is then designed using Lyapunov techniques and incorporates a switching component to ensure robustness with respect to all admissible perturbations within the prescribed bounds. The effectiveness and versatility of the proposed framework are demonstrated through numerical studies involving heterogeneous oscillators, traffic systems, and robotic swarms operating in partially unknown environments.

\subsection{Related Work and Main Contributions}

Density control of large-scale ensembles has been addressed through several methodological paradigms. 
In mean-field optimal control, the problem is formulated as a dynamic optimization, where cost functionals are defined in terms of the population density and are constrained by PDE models of density evolution \cite{fornasier2014mean}. 
Recent developments include sparse control in the presence of additive noise \cite{ascione2023mean} and extensions to multi-population systems \cite{albi2022mean, albi2024kinetic}. 
While these approaches establish well-posedness of optimal solutions, they typically do not yield explicit feedback control laws.

Related optimization-based formulations also arise in optimal transport \cite{villani2021topics} and Schr\"odinger bridge problems \cite{christian2014Asurvey}, where Wasserstein-type distances are minimized. 
Within this context, explicit feedback density control solutions have been derived in simplified linear-quadratic settings \cite{chen2023density}, and more recent efforts have explored distributed optimization-based implementations \cite{brumali2025distributed}. 

An alternative line of research focuses on the direct design of macroscopic feedback control laws with theoretical convergence guarantees \cite{maffettone2022continuification, nikitin2021continuation, fueyo2025continuation}. 
These works develop model-based controllers for first- and second-order agent dynamics, with extensions to leader--follower architectures \cite{maffettone2025leader} and distributed implementations \cite{dilorenzo2025distributed}. 
However, these approaches generally assume nominal agent dynamics and do not explicitly address robustness with respect to heterogeneity and unmodeled disturbances. 
In \cite{maffettone2023continuification2}, an upper bound on the steady-state tracking error under perturbations is derived, but no feedback mechanism is designed to actively reject such disturbances.
Beyond the multi-agent setting, sliding-mode techniques have been extended to infinite-dimensional systems to achieve disturbance rejection in parabolic PDEs. In particular, in~\cite{cristofaro2019robust}, a regularized infinite-dimensional sliding-mode controller is designed for quasilinear reaction--diffusion equations subject to unknown bounded perturbations, achieving robust practical stability with accuracy controlled by the regularization parameter. While that approach addresses a different class of systems (reaction--diffusion equations with Dirichlet boundary conditions), the underlying regularization methodology for handling discontinuous control actions in infinite-dimensional settings is adopted in the present work to establish well-posedness of the proposed control law.

In contrast to existing macroscopic density control approaches, we explicitly account for bounded unknown drift terms by constructing continuum upper and lower bounding density dynamics and designing a single macroscopic feedback law that guarantees asymptotic stabilization of the density error for all admissible perturbations within the prescribed bounds.

{The main contributions of this work are:
\begin{itemize}
\item the derivation of continuum upper and lower bounding density models for large-scale stochastic multi-agent systems affected by unknown but bounded drift terms, capturing both heterogeneity and environmental disturbances;
\item the design of a Lyapunov-based macroscopic feedback controller that guarantees global exponential convergence of the density tracking error in $\mathcal{L}^2$ for the bounding dynamics, and hence for the nominal population via comparison arguments;
\item an extensive numerical validation demonstrating performance and adaptability across heterogeneous oscillators, traffic systems, and robotic swarms operating in partially unknown environments.
\end{itemize}}

{The proposed strategy follows a continuification paradigm \cite{maffettone2022continuification, maffettone2024mixed}, whereby control is designed at the level of the density evolution equation and subsequently discretized to generate microscopic control inputs for individual agents.}

The rest of this paper is organized as follows. 
Sections~\ref{sec:math_modeling} and~\ref{sec:problem_statement} introduce the mathematical modeling framework and formulate the control problem in a one-dimensional setting. 
Section~\ref{sec:control_design} presents the proposed control architecture and the associated stability analysis. 
Section~\ref{sec:1D_Validation} reports numerical validations in representative one-dimensional scenarios. 
Section~\ref{sec:2D_derivation_and_validation} extends the framework to multidimensional domains and presents additional numerical studies. {In Section~\ref{sec:applications}, we showcase the versatility of our control framework for two relevant applications, traffic control in a ring domain, and swarm robotics in the presence of unknown environments.}  
Finally, Section~\ref{sec:Discussion} concludes the paper and outlines future research directions.

\section{Mathematical Modeling}
\label{sec:math_modeling}

We consider a population of $N$ stochastic dynamical systems (agents) evolving over a one-dimensional spatial domain $\Omega \subset \mathbb{R}$. 
Let $x_i(t) \in \Omega$ denote the position of agent $i \in \{1,\dots,N\}$ at time $t$, and let $u_i(t)$ be a control input acting on its velocity. 
A nominal first-order stochastic model for each agent is given by
\begin{equation}\label{eq:single_integrator}
    \mathrm{d}x_i(t) = u_i(t)\,\mathrm{d}t + \sqrt{2D}\,\mathrm{d}W_i(t), \quad i \in \{1,\dots,N\},
\end{equation}
where $W_i(t)$ is a standard one-dimensional Wiener process and $D \geq 0$ is the diffusion coefficient.

Model~\eqref{eq:single_integrator} assumes that the drift of each agent is entirely determined by the control input. 
In practice, however, agent dynamics are typically affected by unmodeled internal behaviors, interactions with the environment, or external disturbances. 
To account for such effects, we adopt the more general stochastic dynamics
\begin{multline}\label{eq:micro_system}
    \mathrm{d}x_i(t) = \bigl[u_i(t) + g_i(t, X(t))\bigr]\mathrm{d}t + \sqrt{2D}\,\mathrm{d}W_i(t), 
    \quad \\
    i \in \{1,\dots,N\},
\end{multline}
where $X(t) = [x_1(t),\dots,x_N(t)]^\top \in \Omega^N$ collects the states of all agents, and 
$g_i : \mathbb{R}_{\geq 0} \times \Omega^N \to \mathbb{R}$ represents unknown drift terms capturing unmodeled dynamics, inter-agent interactions, and environmental perturbations. 
This formulation naturally accommodates heterogeneous agent behaviors, since the functions $g_i$ may differ across agents.

When the number of agents is sufficiently large, it is convenient to describe the collective organization in terms of a macroscopic density function 
$\rho : \Omega \times \mathbb{R}_{\geq 0} \to \mathbb{R}_{\geq 0}$, 
such that for any measurable subset $\mathcal{A} \subset \Omega$, the integral $\int_{\mathcal{A}} \rho(x,t)\,\mathrm{d}x$ represents the fraction of agents located in $\mathcal{A}$ at time $t$. 
The evolution of $\rho$ will be described by a continuum model derived from the microscopic dynamics~\eqref{eq:micro_system} in the next sections.

\section{Problem Statement}
\label{sec:problem_statement}

Our objective is to regulate the macroscopic collective behavior of the population. 
Specifically, we seek control laws $u_i(t)$ such that, starting from arbitrary initial conditions $x_i(0)=x_{i,0}$, $i\in\{1,\dots,N\}$, corresponding to an initial density profile $\rho(x,0)=\rho_0(x)$, the population density asymptotically converges to a prescribed, time-invariant desired density $\rho^{\mathrm d}:\Omega \to \mathbb{R}_{\geq 0}$ satisfying
\begin{align}
\int_{\Omega} \rho^{\mathrm d}(x)\,\mathrm{d}x = 1.
\end{align}
Formally, the control objective is
\begin{equation}
    \label{eq:control_objective}
    \lim_{t \to \infty} \left\Vert \rho^{\mathrm d}(\cdot)-  \rho(\cdot,t)  \right\Vert_{2} = 0,
\end{equation}
where $\|\cdot\|_{2}$ denotes the $\mathcal{L}^2(\Omega)$ norm. The dot notation for the arguments of the functions involved in the norm highlights the variable with respect to which the norm is computed.

We introduce the following assumptions.
\begin{assumption}
\label{ass:omega_boundary_conditions}
The domain $\Omega$ is closed, simply connected, and its boundary, namely $\partial \Omega$, is locally Lipschitz. Boundary conditions are chosen so that agents cannot leave the domain, e.g., periodic or reflective (no-flux) boundary conditions.
\end{assumption}

\begin{remark}
Since $\Omega \subset \mathbb{R}$ is closed and simply connected, it can be equivalently written as $\Omega = [-a,a]$ for some $a>0$. 
Assuming symmetry with respect to the origin simplifies notation without loss of generality.
\end{remark}

\begin{assumption}
\label{ass:g_bounded}
The unknown drift terms affecting the agents are uniformly bounded, i.e.,
\[
|g_i(t,X)| \le K, \quad \forall\, t \ge 0,\ \forall\, X \in \Omega^N,\ \forall\, i \in \{1,\dots,N\},
\]
for some known constant $K>0$.
\end{assumption}

\begin{assumption}
\label{ass:desired_density}
The desired density satisfies $\rho_x^{\mathrm d},\, \rho_{xx}^{\mathrm d} \in \mathcal{L}^\infty(\Omega)$.
Equivalently, the essential suprema $\|\rho_x^{\mathrm d}\|_\infty$ and $\|\rho_{xx}^{\mathrm d}\|_\infty$ are finite.
\end{assumption}

Under Assumption~\ref{ass:omega_boundary_conditions}, the macroscopic density dynamics are mass conservative. 
Consequently, the total population mass satisfies
\begin{equation}
\label{eq:mass_conserved}
\int_{\Omega} \rho(x,t)\,\mathrm{d}x = 1, \qquad \forall\, t \ge 0,
\end{equation}
which is consistent with the normalization of the desired density.

\section{Control design}
\label{sec:control_design}
In this section, we design a control law to solve \eqref{eq:control_objective}. 
Our approach consists of four main steps. 
First, we construct stochastic upper and lower bounding systems for the uncertain microscopic dynamics \eqref{eq:micro_system}. 
Second, we derive continuum macroscopic models describing the spatio--temporal evolution of the densities associated with such bounding systems. 
Third, we design a robust macroscopic feedback law for the resulting density dynamics. 
Finally, the macroscopic control field is discretized and spatially sampled to generate microscopic control inputs for the individual agents. This procedure follows the continuification paradigm described in \cite{maffettone2022continuification}

\subsection{Bounding system}\label{sec:bounding_sys}
Under Assumption~\ref{ass:g_bounded}, the comparison lemma for stochastic differential equations \cite[Chapter~5.3]{skorokhod1982studies} can be applied to the microscopic dynamics \eqref{eq:micro_system}. 
In particular, suppose that $\underline x_i(0) \le x_i(0) \le \bar x_i(0)$ for all $i\in\{1,\dots,N\}$, and consider the auxiliary stochastic systems
\begin{subequations}\label{eq:bounding_sys_micro}
\begin{align}
    \mathrm{d} \underline{x}_i(t) &= \big[u_i(t) - K\big]\mathrm{d}t + \sqrt{2D}\,\mathrm{d}W_i(t), \quad i \in \{1,\dots,N\}, \label{eq:lowerbound_micro}\\
    \mathrm{d} \bar x_i(t) &= \big[u_i(t) + K\big]\mathrm{d}t + \sqrt{2D}\,\mathrm{d}W_i(t), \quad i \in \{1,\dots,N\}. \label{eq:upperbound_micro}
\end{align}
\end{subequations}
Then, for all $t \ge 0$ and all $i\in\{1,\dots,N\}$,
\begin{equation}\label{eq:comparison_lemma}
    \underline{x}_i(t) \le x_i(t) \le \bar x_i(t).
\end{equation}

For large populations, the collective behavior of the bounding systems can be described in terms of their associated macroscopic densities 
$\underline{\rho},\, \bar\rho : \Omega \times \mathbb{R}_{\ge 0} \to \mathbb{R}_{\ge 0}$, 
which will be shown to satisfy suitable continuum advection--diffusion equations (Fokker-Planck) in the next subsection.

\begin{remark}
If, under a given control law, both bounding densities satisfy
$\underline{\rho}(\cdot,t)\to\rho^\mathrm{d}$ and $\bar\rho(\cdot,t)\to\rho^\mathrm{d}$ as $t\to\infty$, then convergence of the nominal density $\rho(\cdot,t)\to\rho^\mathrm{d}$ follows directly from \eqref{eq:comparison_lemma}.
\end{remark}

\subsection{Macroscopic model}
For very large populations (i.e., as $N \to \infty$), following \cite{nikitin2021continuation,bernoff2011primer,leverentz2009asymptotic}, 
the spatio--temporal evolution of the densities associated with the bounding systems \eqref{eq:bounding_sys_micro} can be described by the advection--diffusion equations
\begin{subequations}\label{eq:continuum_sys}
\begin{align}
        \underline{\rho}_t(x, t) + \big[\underline{\rho}(x, t)\big(U(x,t) - K \big)\big]_x &= D \underline{\rho}_{xx}(x, t), \label{eq:lower_bounding_pde}\\
        \bar{\rho}_t(x, t) + \big[\bar\rho(x, t)\big(U(x,t) + K \big)\big]_x &= D \bar\rho_{xx}(x, t), \label{eq:upper_bounding}
\end{align}
\end{subequations}
where $x \in \Omega$ and $U : \Omega \times \mathbb{R}_{\ge 0} \to \mathbb{R}$ denotes the macroscopic velocity field corresponding to the microscopic inputs $u_i$.

For notational convenience, both bounding dynamics can be compactly written as
\begin{equation}\label{eq:bounding_pde}
    \hat{\rho}_t(x, t) + \big[\hat\rho(x, t)\big(U(x,t) \pm K \big)\big]_x = D \hat\rho_{xx}(x, t),
\end{equation}
where the $\pm$ sign corresponds to the upper and lower bounding systems, respectively.
We consider initial conditions $\hat{\rho}(x, 0)=\hat{\rho}_0(x)$ taken in the Sobolev space $H^2(\Omega)\equiv W^{2,2}(\Omega)$, that is, the space of functions with a $\mathcal{L}^2(\Omega)$ second-order spatial derivative. This implies interpreting derivatives in \eqref{eq:bounding_pde} in a weak sense. However, sampling initial conditions in $\mathcal{C}^2(\Omega)$, and assuming a sufficiently smooth control action allow to interpret derivatives in a strong sense \cite{quarteroni2009numerical}.

To enforce mass conservation as required by Assumption~\ref{ass:omega_boundary_conditions}, 
we impose zero-flux boundary conditions on \eqref{eq:bounding_pde}, namely
\begin{equation}\label{eq:no_flux_at_boundaries}
\begin{cases}
    \big[D\hat{\rho}_x(x, t)\pm K\hat{\rho}(x, t)\big]_{-a}^a = 0,\\
    \big[\hat{\rho}(x, t)U(x, t)\big]_{-a}^a = 0.
\end{cases}
\end{equation}

\begin{remark}
Boundary conditions \eqref{eq:no_flux_at_boundaries} are satisfied, for instance, under periodic boundary conditions on both $\hat\rho$ and $U$, or under reflective (no-flux) boundary conditions, i.e.,
\begin{align}
\begin{cases}
    \big[D\hat{\rho}_x(x, t)\pm K\hat{\rho}(x, t)\big]\big|_{x=\pm a} = 0,\\
    \big[\hat{\rho}(x, t)U(x, t)\big]\big|_{x=\pm a} = 0.
\end{cases}
\end{align}
\end{remark}

\subsection{Macroscopic control design}

We now design a macroscopic feedback law for the bounding density dynamics \eqref{eq:bounding_pde}. 
Let us define the density tracking error as
\begin{equation}\label{eq:error_def}
    e(x, t) := \rho^{\mathrm{d}}(x) - \hat{\rho}(x,t).
\end{equation}
Since $\rho^{\mathrm{d}}$ is time-invariant, differentiating \eqref{eq:error_def} in time and using \eqref{eq:bounding_pde} yields
\begin{equation}\label{eq:error_dyn_rho}
    e_t(x, t)
    = \big[\hat{\rho}(x, t)U(x, t)\big]_x \pm K\hat{\rho}_x(x, t) - D\hat{\rho}_{xx}(x, t).
\end{equation}
Recalling that $\hat{\rho} = \rho^{\mathrm{d}} - e$, the error dynamics can be rewritten as
\begin{equation}\label{eq:error_dynamics}
    e_t(x, t) = q(x, t) \pm K\rho^{\mathrm{d}}_x(x) \mp K e_x(x,t)
    - D \rho^{\mathrm{d}}_{xx}(x) + D e_{xx}(x,t),
\end{equation}
where we define the control-related term
\begin{equation}\label{eq:q_definition}
    q(x, t) := \left[\hat{\rho}(x,t)U(x,t)\right]_x .
\end{equation}

From the zero-flux boundary conditions \eqref{eq:no_flux_at_boundaries}, the error satisfies
\begin{equation}\label{eq:err_boundary_cod}
    \begin{cases}
    \big[\mp K e(x,t) + D e_x(x,t)\big]_{-a}^a = 0,\\
    \big[\hat{\rho}(x,t) U(x,t)\big]_{-a}^a = 0.
    \end{cases}
\end{equation}
We assume that the desired density satisfies compatible boundary conditions\footnote{For periodic domains, it suffices that $\rho^{\mathrm{d}}$ is periodic. For reflective boundaries, it must hold that $\mp K \rho^{\mathrm{d}} + D \rho^{\mathrm{d}}_x|_{-a}^a = 0$, which is satisfied, for instance, if $\rho^{\mathrm{d}}=\rho^{\mathrm{d}}_x=0$ at the boundary.}.  
Hence, the control objective reduces to designing $q$ such that $\lim_{t\to\infty}\|e(\cdot,t)\|_2 = 0$.

\begin{theorem}[Global exponential macroscopic convergence]\label{th:convergence_1d}
Choose
\begin{equation}\label{eq:q}
    q(x, t) = -k_p e(x, t) - k_s(t)\,\mathrm{sign}(e(x, t)) + \alpha(t),
\end{equation}
where $k_p>0$, $\alpha(t)$ is any bounded function of time, and $k_s(t)$ satisfies
\begin{equation}\label{eq:kst}
    k_s(t) > A + K\|e_x(\cdot, t)\|_\infty , \qquad \forall t\ge 0,
\end{equation}
with
\begin{equation}\label{eq:A}
    A := D \|\rho^{\mathrm{d}}_{xx}(\cdot)\|_\infty + K \|\rho^{\mathrm{d}}_x(\cdot)\|_\infty .
\end{equation}
Then, the solution of \eqref{eq:error_dynamics} satisfies
\begin{equation}
    \Vert e(\cdot, t)\Vert_2^2 \le \Vert e(\cdot,0)\Vert_2^2 \,\mathrm{exp}(-k_p t),
\end{equation}
i.e., the density tracking error converges globally and exponentially to zero in $\mathcal{L}^2(\Omega)$.
\end{theorem}

\begin{proof}
Let us introduce the candidate Lyapunov functional

\begin{equation}
    V(t)  = \frac{1}{2}\Vert e(\cdot, t)\Vert_2^2 = \frac{1}{2} \int_\Omega e^2(x, t) \, \mathrm{d}x.
    \label{eq:Lyap_function}
\end{equation}
Differentiating $V$ in time and using \eqref{eq:error_dynamics}, we get (dropping functional dependencies for simplicity)
\begin{multline}
    V_t = \int_\Omega e e_t \,\mathrm{d}x = \int_\Omega eq\, \,\mathrm{d}x 
    \pm K \int_\Omega  e \rho^{\mathrm{d}}_x  \,\mathrm{d}x \\- D \int_\Omega e \rho ^{\mathrm{d}}_{xx} \, \mathrm{d}x 
    + D \int_\Omega  e e_{xx}\, \mathrm{d}x \mp K \int_\Omega  e e_x \, \mathrm{d}x. 
    \label{eq:dot_v}
\end{multline}
To investigate the sign of $V_t$, we focus separately on its terms. First, we note that 
\begin{multline}\label{eq:bound_general_case}
    D \int_\Omega e  e_{xx} \, \mathrm{d}x \mp K\int_\Omega e  e_{x} \, \mathrm{d}x 
    = D\left[ee_x\right]_{-a}^{a} - D\int_\Omega e_x^2 \, \mathrm{d}x 
    \\\mp \frac{K}{2} \int_\Omega (e^2)_x \, \mathrm{d}x  \leq \left[e\left(D e_x \mp \frac{K}{2}e\right) \right]_{-a}^a,
\end{multline}
where we applied integration by parts on the first term (recalling that $\int_\Omega (e_x)^2 dx \leq 0$) and recalled that $(e^2)_x = 2ee_x$. By adding and subtracting $[\pm K/2 \,e^2]_{-a}^a$ in the rightmost term of \eqref{eq:bound_general_case}, we recover
\begin{multline}\label{eq:boud_gigante}
    D \int_\Omega e  e_{xx} \, \mathrm{d}x \mp K\int_\Omega e  e_{x} \, \mathrm{d}x 
    \leq  \left[e\left(D e_x \mp K e\right) \right]_{-a}^a \\\mp \frac{K}{2} \left[e^2\right]_{-a}^a 
    = \mp \frac{K}{2} \left[e^2\right]_{-a}^a, 
\end{multline}
where we applied boundary conditions -- see \eqref{eq:err_boundary_cod}. Furthermore, it holds
\begin{multline}\label{eq:bound_e^2}
    \mp \frac{K}{2} \left[e^2\right]_{-a}^a = \pm K \int_\Omega ee_x \, \mathrm{d}x   \leq K \int_\Omega \left\vert ee_x\right\vert \, \mathrm{d}x= K\Vert e e_x\Vert_1\\
    \leq K\Vert e \Vert_1 \Vert e_x \Vert_\infty = K\Vert e_x\Vert_\infty\int_\Omega \vert e\vert\,\mathrm{d}x,
\end{multline}
where we applied the definition of $\mathcal{L}^1(\Omega)$ norm and the Holders' inequality \cite{axler2020measure}. Hence, combining \eqref{eq:bound_general_case}, \eqref{eq:boud_gigante} and \eqref{eq:bound_e^2},  it holds that
\begin{equation}\label{eq:final_bound_multiple terms}
    D \int_\Omega e  e_{xx} \, \mathrm{d}x \mp K\int_\Omega e  e_{x} \, \mathrm{d}x
    \leq K\Vert e_x\Vert_\infty\int_\Omega \vert e\vert\,\mathrm{d}x.
\end{equation}

To study the remaining terms in \eqref{eq:dot_v} we recall that, for any function $h(x)\in\mathcal{L}^\infty(\Omega)$, we can write 
\begin{equation}
    \pm \int_\Omega  e h \, \mathrm{d}x \leq \left|\int_\Omega  e h \, \mathrm{d}x \right|
    \leq \int_\Omega  \left|e h \right| \, \mathrm{d}x = \Vert e h(\cdot)\Vert_1. 
\end{equation}
Additionally, applying Holder's inequality, it holds
\begin{equation}
    \pm \int_\Omega  e h \, \mathrm{d}x \leq \Vert e h\Vert_1
    \leq \Vert e\Vert_1\Vert h\Vert_\infty = \Vert h\Vert_\infty \int_\Omega\vert e\vert\mathrm{d}x .
    \label{eq:bound_holder}
\end{equation}
Given the bound in equation \eqref{eq:bound_holder}, and being  $K$ and $D$ positive constants, yields 
\begin{multline}
    \pm K \int_\Omega  e \rho^\mathrm{d}_x \, \mathrm{d}x - D \int_\Omega e \rho^\mathrm{d}_{xx} \, \mathrm{d}x  \leq K \Vert \rho^\mathrm{d}_x\Vert_\infty \int_\Omega  \vert e\vert\,\mathrm{d}x \\+ D \Vert \rho ^\mathrm{d}_{xx}\Vert_\infty\int_\Omega \vert e\vert\,\mathrm{d}x  = A \int_\Omega \vert e\vert\,\mathrm{d}x ,
    \label{eq:positive_terms}
\end{multline}
where $A$ (see \eqref{eq:A}) is a positive constant that depends on the reference density, the maximum amplitude of the unknown internal dynamics and the diffusion coefficient.

Given the bounds in \eqref{eq:final_bound_multiple terms} and \eqref{eq:positive_terms}, and applying them to \eqref{eq:dot_v},  yields   
\begin{align}\label{eq:Vt}
     V_t  \leq \int_\Omega eq\, \mathrm{d}x 
     + \left(A + K\Vert e_x \Vert_\infty\right) \int_\Omega \vert e\vert \,\mathrm{d}x.
\end{align}
We fix $q$ as in \eqref{eq:q} and, substituting it into \eqref{eq:Vt}, we get
\begin{multline}
    V_t \leq -k_p\int_\Omega e^2 \,\mathrm{d}x + (A + K\Vert e_x \Vert_\infty - k_s) \int_\Omega \vert e\vert \,\mathrm{d}x  + \alpha\int_\Omega e \,\mathrm{d}x \\= -k_p\int_\Omega e^2 \,\mathrm{d}x + (A + K\Vert e_x \Vert_\infty - k_s) \int_\Omega \vert e\vert \,\mathrm{d}x,
\end{multline}
where we used that $\int_\Omega e\,\mathrm{d}x = 0$ due to boundary conditions ensuring mass conservation. Then, choosing the time-varying gain $k_s$ as in \eqref{eq:kst}, implies $V_t < -k_p V$, proving the claimed exponential convergence. The rate of convergence is upper bounded by $k_p$.
\end{proof}

\begin{remark}\label{rem:wellposedness}
    The control action in~\eqref{eq:q} is discontinuous due to the $\mathrm{sign}(\cdot)$ term and solutions of the closed-loop dynamics should  therefore be interpreted in the Filippov sense \cite[Chapter 2]{filippov88}. Well-posedness is addressed in Sec.~\ref{sec:well_posedness} by introducing a smooth $\epsilon$-parameterized regularization of~\eqref{eq:q} following the approach of~\cite{cristofaro2019robust,Orlov2000} for infinite-dimensional switching control of parabolic PDEs. 
\end{remark}

\begin{remark}[Recovery of the velocity field]\label{rem:alpha}
Given the definition \eqref{eq:q_definition}, the macroscopic velocity field $U$ can be recovered from $q$ by spatial integration. 
The additive function $\alpha(t)$ does not affect stability but must be selected to ensure that the resulting flux satisfies boundary conditions \eqref{eq:no_flux_at_boundaries}. 
Details are provided in Section~\ref{sec:choiceofalpha}.
\end{remark}

\begin{remark}[Periodic boundary conditions]
Under periodic boundary conditions, the Lyapunov analysis in Theorem \ref{th:convergence_1d} simplifies as the boundary terms in \eqref{eq:bound_general_case} vanish, making 
\begin{align}
    D \int_\Omega e  e_{xx} \, \mathrm{d}x \mp K\int_\Omega e  e_{x} \, \mathrm{d}x \leq 0.
\end{align} 
In this case, it suffices to select a time invariant gain $k_s > A$ to guarantee global stability.
\end{remark}

\subsection{Closed-loop analysis and extensions}
We now discuss structural properties of the closed-loop system 
induced by the control law~\eqref{eq:q}, establishing well-posedness 
via regularization, and presenting extensions to tracking problems 
with time-varying desired densities.
{

\subsubsection{Well-posedness of the closed loop dynamics}\label{sec:well_posedness}

We study the well-posedness of the closed-loop solution of the error dynamics in \eqref{eq:error_dynamics} under the control action given by \eqref{eq:q}, adopting the same approach as in \cite[Section 4]{cristofaro2019robust}. Specifically, for any $\epsilon > 0$, let $\mathrm{sign}_\epsilon : \mathbb{R} \to [-1,1]$ be a smooth monotone approximation of the sign function (e.g., $\mathrm{sign}_\epsilon(s) = \tanh(s/\epsilon)$) such that $\|L_\epsilon\|_2 \to 0$ as $\epsilon \to 0$ with $L_\epsilon:=\mathrm{sign}_\epsilon(\cdot) - \mathrm{sign}(\cdot)$. Let us construct the regularized control input
\begin{align}\label{eq:regularize_control}
    q_\epsilon(x, t) = -k_p\, e(x, t) - k_s(t)\, \mathrm{sign}_{\epsilon}[e(x, t)] + \alpha(t).
\end{align}
Following \cite[Definition~1]{cristofaro2019robust}, the 
closed-loop system~\eqref{eq:error_dynamics} 
with~\eqref{eq:q} admits a generalized solution if the 
regularized problem possesses a family of strong solutions 
converging in $\mathcal{L}^2(\Omega)$ to a limit as 
$\epsilon \to 0$.

In the following Theorem we study the stability properties of \eqref{eq:error_dynamics} under the regularized control action \eqref{eq:regularize_control}. 

\begin{theorem}[Stability under a regularized control action]
    Let $m_\epsilon(t)= \sqrt{2} k_s(t) \lVert L_\epsilon \rVert_2$ with $k_s$ chosen as in \eqref{eq:kst}.  Under the regularized control action \eqref{eq:regularize_control}, the solution of the error dynamics, say $e_\epsilon(t)$, satisfies 
    \begin{equation}
        \lim_{t\to\infty}\lVert e_\epsilon(\cdot, t) \lVert_2 \leq \frac{1}{k_p}\lim_{t\to\infty}m_\epsilon(t).
    \end{equation}
\end{theorem}
\begin{proof}
Let us rewrite 
\begin{align}\label{eq:F_decomposition}
    \mathrm{sign}_\epsilon[e(x, t)] = \mathrm{sign}[e(x, t)] - L_\epsilon[e(x, t)],
\end{align}
where $L_\epsilon$ denotes the mismatch between the sign and its regularized counterpart. 

Substituting~\eqref{eq:regularize_control} into~\eqref{eq:Vt} yields
\begin{multline}\label{eq:Vt_smooth}
    V_t \leq -k_p\, V + \big(A + K\lVert e_x \rVert_\infty - k_s\big) \int_\Omega \lvert e \rvert\,\mathrm{d}x \\
    + k_s \int_\Omega e\, L_\epsilon(e)\,\mathrm{d}x.
\end{multline}
Bounding the last term using H$\ddot{\mathrm{o}}$lder's inequality and recalling that $\big(A + K\lVert e_x \rVert_\infty - k_s\big) \int_\Omega \lvert e \rvert\,\mathrm{d}x\leq0$ under our choice of $k_s$, we write
\begin{align}\label{eq:Vt_smooth_bound}
    V_t \leq -k_p\, V + \sqrt{2} k_s \lVert L_\epsilon \rVert_2 \sqrt{V} 
\end{align}
By defining $z = \sqrt{V}$, we can study the dynamics of the bounding system of \eqref{eq:Vt_smooth_bound}, that is 
\begin{align}\label{eq:ks_smooth}
    z_t(t) = \frac{1}{2}\left[ -k_p z(t) + m_\epsilon(t)\right].
\end{align}
Being $k_p>0$ by construction, it holds
\begin{align}
    \lim_{t\to\infty}z(t) = \frac{1}{k_p}\lim_{t\to\infty}m_\epsilon(t).
\end{align}
\end{proof}
\begin{remark}
    Since $\|L_\epsilon\|_2 \to 0$ as $\epsilon \to 0$, it holds that
    \begin{align}
        \lim_{\epsilon \to 0} m_\epsilon(t) = 0.
    \end{align}
    This ensures the existence of a generalized solution of~\eqref{eq:error_dynamics} under~\eqref{eq:q} in the sense of~\cite[Definition~1]{cristofaro2019robust}. Moreover, the regularized control action~\eqref{eq:regularize_control} suppresses chattering phenomena arising from the discontinuity of the $\mathrm{sign}(\cdot)$ term.
\end{remark}

\subsubsection{Extension to tracking problems} \label{sec:Tracking_problem}
The proposed design readily extends to tracking problems with time-varying desired densities $\rho^{\mathrm{d}}(x,t)$, provided that the total mass is conserved, i.e.,
$\left(\int_\Omega \rho^{\mathrm{d}}(x,t)\,dx\right)_t = 0$. 
In this case, adding a feedforward term $-\rho^{\mathrm{d}}_t$ to $q$ in \eqref{eq:q} preserves the stability properties of the closed-loop system.

}

\subsection{\texorpdfstring{Choice of $\alpha$}{Choice of alpha}}\label{sec:choiceofalpha}
\begin{figure*}[t]
    \centering
    \includegraphics[width=1\linewidth]{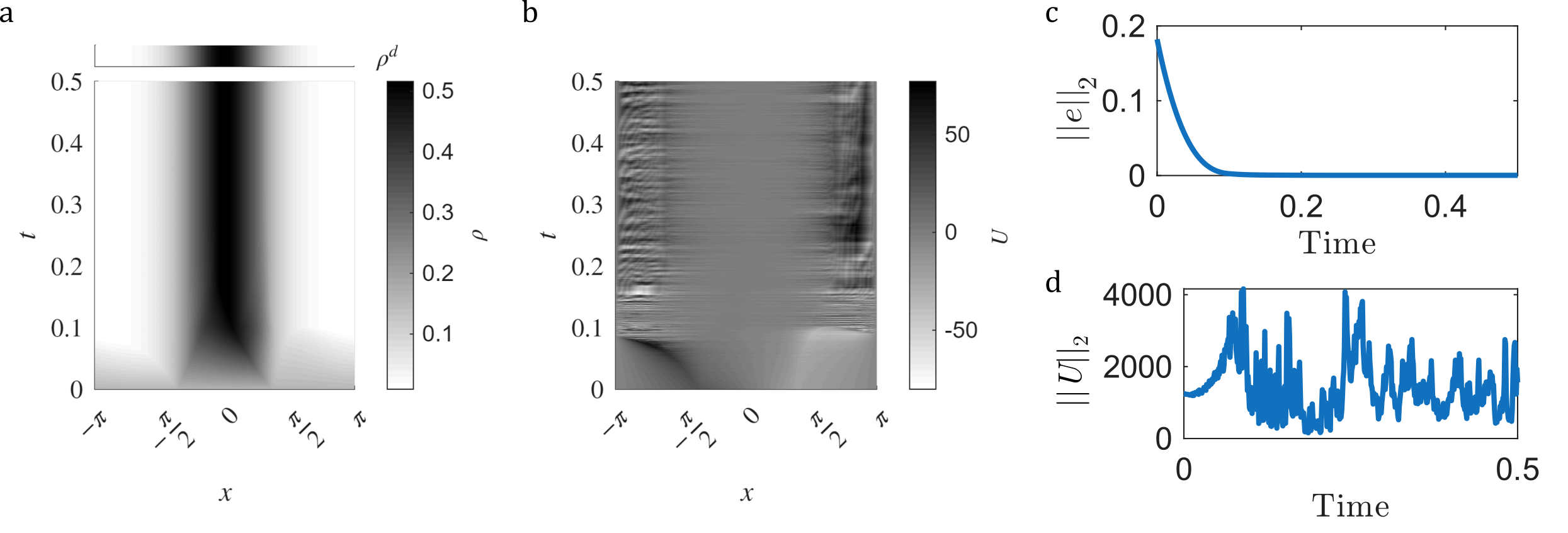}
    \caption{\textbf{Validation of the controller in a 1D macroscopic setting}. \textbf{a}. Evolution in time and space (bottom panel) of the density described by Equation \eqref{eq:upper_bounding}. The upper panel shows the desired density using the same colormap used for the density portrayed in the bottom panel. 
    \textbf{b}. Evolution in time and space of the velocity field $U(x,t)$ generated by the controller.
    \textbf{c}. $\mathcal{L}^2$ norm of the control error over time.
    \textbf{d}. $\mathcal{L}^2$ norm of the velocity field induced by the control action in time.
    }
    \label{fig:1D_Validation_macroscopic}
\end{figure*}

The additive term $\alpha(t)$ in \eqref{eq:q} does not
affect stability, but must be selected so that the resulting velocity field satisfies the zero-flux boundary conditions \eqref{eq:no_flux_at_boundaries}. 
Since $q = (\hat{\rho} U)_x$, the control flux can be recovered by spatial integration as
\begin{align}\label{eq:U_from_q}
    \hat{\rho}(x, t) U(x, t)
    &= -k_p \int e(x, t)\,\mathrm{d}x\\
    &- k_s(t) \int \!\big[\mathrm{sign}(e(x, t))+\alpha(t)\big] \mathrm{d}x
    + B(t),
\end{align}
where $B(t)$ is an arbitrary function of time coming from the integration in space.

We now specify $\alpha(t)$ for the two classes of boundary conditions considered in this paper.

\paragraph{Periodic boundary conditions}
Under periodic boundary conditions, mass conservation implies $\int_\Omega e(x,t)\,\mathrm{d}x = 0$. 
Moreover, the following classical result holds \cite[Ch.~1, Th.~1.6]{katznelson2004introduction}.

\begin{lemma}\label{lemma:periodicity}
If $f:\Omega\to\mathbb{R}$ is periodic and $\int_\Omega f(x)\,\mathrm{d}x = 0$, then its indefinite integral is also periodic.
\end{lemma}

Therefore, the first term on the right-hand side of \eqref{eq:U_from_q} is periodic. 
To make the second term periodic as well, it suffices to impose
\begin{equation}\label{eq:alpha_periodic}
    \alpha(t) = \frac{k_s(t)}{2a}\int_\Omega \mathrm{sign}(e(x,t))\,\mathrm{d}x ,
\end{equation}
which guarantees $\int_\Omega [k_s \mathrm{sign}(e)+\alpha]\,dx = 0$. 
Since $\mathrm{sign}(\cdot)$ is bounded, $\alpha(t)$ is bounded for all $t$. 
The term $B(t)$ does not affect periodicity and can be set arbitrarily.

\paragraph{Reflective boundary conditions}
For reflective boundaries, \eqref{eq:no_flux_at_boundaries} requires that the flux vanishes at $x=\pm a$. 
Rewriting \eqref{eq:U_from_q} as
\begin{align}\label{eq:ufromq_reflective}
    \hat{\rho}(x, t) U(x, t)
    = G(x,t) + \alpha(t)x + B(t),
\end{align}
with
\begin{equation}
    G(x,t) = -k_p \int e(x,t)\,\mathrm{d}x - k_s(t)\int \mathrm{sign}(e(x,t))\,\mathrm{d}x ,
\end{equation}
the unique choice of $\alpha(t)$ and $B(t)$ ensuring zero flux at $x=\pm a$ is
\begin{align}
    \alpha(t) &= \frac{G(-a,t)-G(a,t)}{2a},\\
    B(t) &= -\frac{G(-a,t)+G(a,t)}{2}.
\end{align}

\subsection{Discretization}
We have derived a macroscopic control input $q$ that guarantees convergence of the density dynamics of the bounding systems $\underline{\rho}$ and $\bar{\rho}$ to the desired profile $\rho^\mathrm{d}$. By the comparison argument in Section~\ref{sec:bounding_sys}, this implies convergence of the nominal density $\rho$ associated with the original stochastic multi-agent system~\eqref{eq:micro_system}. Therefore, the macroscopic control law is suitable to regulate the collective behavior of the actual agent population.

From \eqref{eq:U_from_q}, the macroscopic velocity field $U(x,t)$ can be recovered from $q$ by spatial integration. To obtain microscopic control inputs for each agent, we adopt a spatial sampling strategy as in~\cite{maffettone2022continuification}, and apply the velocity field at the current agent locations, namely
\begin{align}\label{eq:spatial_sampling}
    u_i(t) = U(x_i(t), t), \quad i \in \{1,\dots,N\}.
\end{align}
This yields a distributed feedback law where each agent requires only the evaluation of the macroscopic velocity field at its own position.

When implementing the control law on the microscopic dynamics~\eqref{eq:micro_system}, the velocity field $U$ must be computed using the actual density $\rho$, rather than the bounding densities $\hat{\rho}$. In practice, $\rho$ is not available in closed form and must be estimated from the agent positions. Standard techniques such as kernel density estimation or histogram-based approximations can be employed to reconstruct $\rho(x,t)$ in real time from $\{x_i(t)\}_{i=1}^N$.

We remark that the convergence analysis is performed at the macroscopic level, while the microscopic implementation introduces approximation errors due to finite population effects and density estimation. 
{Nevertheless, for sufficiently large $N$ and accurate density reconstruction, the closed-loop multi-agent system behaves sufficiently in accordance with its continuum counterpart, as confirmed by the numerical simulations in Section~\ref{sec:1D_Validation}}.

\begin{figure*}[!t]
    \centering
    \includegraphics[width=1\linewidth]{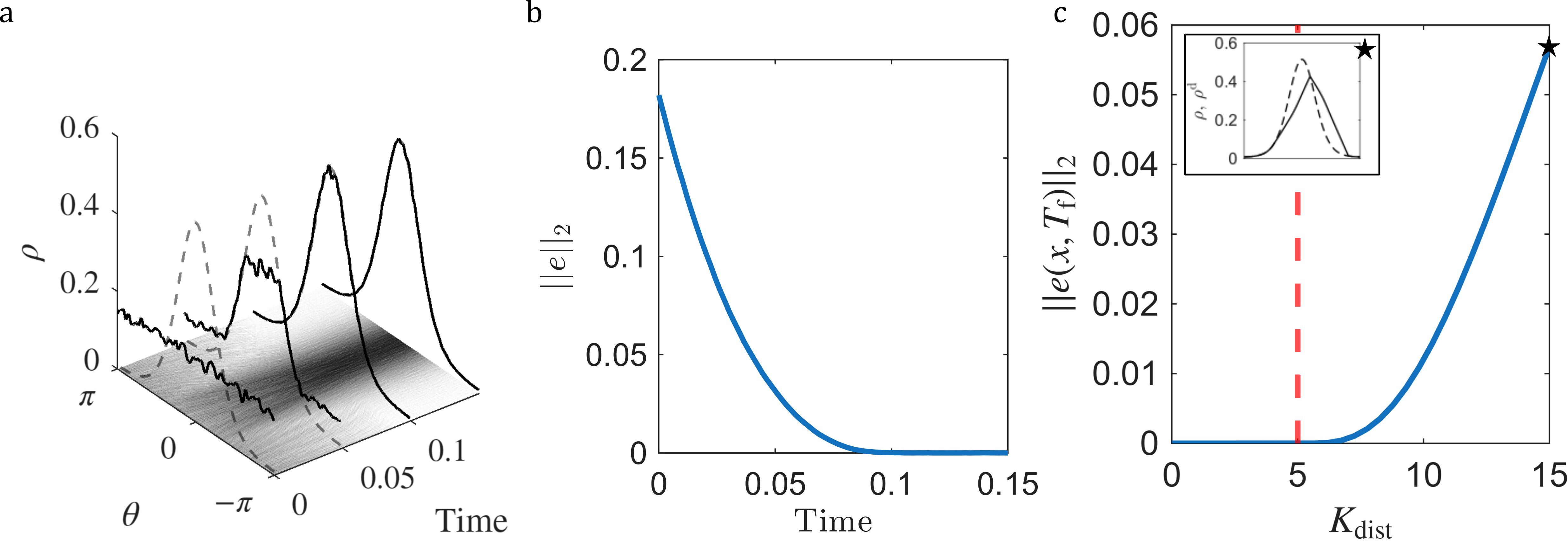}
    \caption{\textbf{Control of an ensemble of heterogeneous stochastic oscillators}. \textbf{a}. Evolution in time an space of all the agents in the ensemble. On the z-axis the estimated (solid) and desired (dashed) densities are displayed in four representative time instants. 
    \textbf{b}. Evolution of $\Vert e\Vert_2$ in time.
    \textbf{c}. Residual steady state error $e(x,T)$ for increasing values of the oscillators heterogeneity $K_{dist}$. The red dashed vertical line represents the level of heterogeneity used to determine the control gains. $T_f$ is the final time instant of the simulation, set as in panels a and b as 0.15 time units. The inset shows the final density density (solid black line) against the reference density (dashed black line).
    }
    \label{fig:1D_Validation_microscopic}
\end{figure*}

\section{Numerical validation} \label{sec:1D_Validation}
We validate our control strategy both at the macroscopic continuum ($N\to\infty$) and at the microscopic level for finite agent populations.
{All numerical simulations were implemented in MATLAB. All the necessary details to reproduce simulations are provided when needed.}

\subsection{Macroscopic validation}
We first consider the continuum setting described by the bounding system \eqref{eq:upper_bounding} in a periodic domain $\Omega=[-\pi,\pi]$. We set $D = 0.1$, $K=5$, $k_p=1$, and $k_s=1.1(D \|\rho^\mathrm{d}_{xx}\|_\infty + K \|\rho^\mathrm{d}_x\|_\infty)$ so as to satisfy condition \eqref{eq:kst}. We set the desired density as a Von-Mieses distribution 
\begin{equation} \label{eq:Von_Mieses_Distribution}
    \rho^\mathrm{d}(x)=\mathrm{exp}
    \,(\kappa  \cos(x-\mu)),
\end{equation}
where, $\kappa = 2$ and $\mu = 0$ represent the concentration and mean.
The PDE in \eqref{eq:upper_bounding} is numerically integrated using a Lax-Friedrichs finite volume scheme \cite{leveque2002finite} with a spatial grid of 200 points and a time step of $\Delta t = 10^{-4}$.
%

As expected, {choosing the macroscopic control input $q$ as in \eqref{eq:q} (and consequently $U$ from its spatial integration)}, the density converges to $\rho^\mathrm{d}$, while the velocity field exhibits rapid oscillations as the error approaches zero to compensate for diffusion and unmodeled dynamics (see Figs.~\ref{fig:1D_Validation_macroscopic}a,~\ref{fig:1D_Validation_macroscopic}b).
Additionally, we quantified the controller performance by analyzing the $\mathcal{L}^2(\Omega)$ norms of the control error and of the induced velocity field.
The error norm decreases monotonically, approaching zero in approximately 0.1 time units (see Fig.~\ref{fig:1D_Validation_macroscopic}c). In contrast, as shown in Fig.~\ref{fig:1D_Validation_macroscopic}d, the velocity field norm—a proxy for control effort—does not vanish even when the error converges to zero. This persistent control activity is due to the switching action required to stabilize the solution in the presence of unmodeled dynamics. 

\subsection{Microscopic validation}
We validate the control strategy at the microscopic level by considering a system of one-dimensional stochastic oscillators with 
heterogeneous natural frequencies, whose dynamics are described in a periodic domain by
\begin{equation} \label{eq:1D_Oscillator}
    \mathrm{d}\theta_i = (\omega_i+u_i) \mathrm{d}t + \sqrt{2D}\mathrm{d}W_i, \quad i\in\{1,\dots, N\}.
\end{equation}
Here, $\omega_i$ is the natural frequency of the $i$-th oscillator, and $D$ is the variance of the noise affecting the oscillator dynamics. 
Note that by defining $x_i=\theta_i$ and $g_i(x_i(t))=\omega_i$, system \eqref{eq:1D_Oscillator} takes the same form as \eqref{eq:micro_system}.
Furthermore, choosing $K = \max_i (\omega_i)$ uniformly bounds the dynamics of all oscillators in the system.

In our numerical experiments, we simulate $N=5000$ agents with the desired density given by the von Mises distribution in \eqref{eq:Von_Mieses_Distribution} with $\kappa=2$ and $\mu=0$. We set $D=0.1$ and draw each $\omega_i$ independently from a uniform distribution $\mathcal{U}(-5,5)$. 
%

Consequently, we set $K=5$ to define the upper and lower bounds for the system dynamics.
The control gains are chosen as $k_p=1$ and $k_s=1.1(D \Vert\rho ^\mathrm{d}_{xx}\Vert_\infty + K \Vert \rho^\mathrm{d}_x\Vert_\infty)$, ensuring that the hypotheses of Theorem~\ref{th:convergence_1d} hold. 
Given the finite population size, { for the computation of the macroscopic control action $q$ in \eqref{eq:q},} we estimate the density by normalizing and filtering (with a Gaussian kernel) the histogram of agent phases on the spatial grid. {Next, the macroscopic control field $U(x,t)$ is recovered from $q$ by spatial integration as described 
in Section~\ref{sec:choiceofalpha}. The control inputs to the oscillators in \eqref{eq:1D_Oscillator} are then obtained by spatially sampling the 
macroscopic velocity field $U(x,t)$ as in \eqref{eq:spatial_sampling}.}
We integrate \eqref{eq:1D_Oscillator} using the Euler-Maruyama scheme with time step $\Delta t=5\cdot10^{-5}$. 

The agent trajectories in Fig.~\ref{fig:1D_Validation_microscopic}a show that the oscillator phases converge to the prescribed desired density $\rho^\mathrm{d}$ after a short transient. 
As in the macroscopic simulation, we assess controller performance by analyzing the $\mathcal{L}^2(\Omega)$ norm of the control error, which decreases monotonically and approaches zero in approximately 0.1 time units (see Fig.~\ref{fig:1D_Validation_microscopic}b).

We also investigate the conservativeness of Theorem~\ref{th:convergence_1d} by running numerical experiments with incorrect estimates of $K$. 
Specifically, we fix $K=5$ for the control design while extracting the natural frequencies from uniform distributions in the interval $[-K_\mathrm{dist},K_\mathrm{dist}]$ (i.e., $\omega_i\sim\mathcal{U}(-K_\mathrm{dist},K_\mathrm{dist})$). We perform a parameter sweep over $K_\mathrm{dist} \in [0,15]$, keeping all other parameters identical to those used in Figs.~\ref{fig:1D_Validation_microscopic}a,~b.
For each experiment, we record the residual error norm $\Vert e(\cdot,T_\mathrm{f})\Vert_2$ at the final simulation time $T_\mathrm{f}$. As shown in Fig.~\ref{fig:1D_Validation_microscopic}c, when $K_\mathrm{dist}\leq K$ (vertical red dashed line), the controller guarantees zero residual error. In contrast, when $K_\mathrm{dist} > K$, a residual steady-state error emerges and increases with the degree of agent heterogeneity. The inset of Fig.~\ref{fig:1D_Validation_microscopic}c shows the final density profile for $K_\mathrm{dist} = 15$, demonstrating qualitative agreement between $\rho$ and $\rho^\mathrm{d}$ despite the model mismatch.

\section{Higher dimensional framework} \label{sec:2D_derivation_and_validation}
In this section, we expand our framework to more general higher-dimensional domains.

\subsection{Mathematical modeling}
In higher-dimensional domains, the dynamics of each agent are governed by
\begin{multline}\label{eq:micro_hd}
    \mathrm{d}\mathbf{x}_i(t) = \left[\mathbf{u}_i(t) + \mathbf{g}_i(t, \mathbf{X}(t))\right] \mathrm{d}t 
    + \sqrt{2D}\, \mathrm{d}\mathbf{W}_i(t), 
\end{multline}
for $i\in\{1,\dots,N\}$, where $\mathbf{x}_i\in\Omega\subset\mathbb{R}^n$ with $n=2,3$ (for instance, if $n=3$, then $\mathbf{x}_i = [x_{i,1}, x_{i,2}, x_{i,3}]^\top$), $\mathbf{u}_i$ is an $n$-dimensional control input, $\mathbf{X}\in\Omega^{Nn}$ collects the positions of all agents at time~$t$, $\mathbf{g}_i$ captures the uncompensated internal dynamics of agent~$i$ and potential environmental disturbances, $D\geq0$ is the diffusion coefficient, and $\mathbf{W}_i$ is an $n$-dimensional standard Wiener process. We adopt the same problem statement introduced in Section~\ref{sec:problem_statement} and, analogously to the one-dimensional setting, require that Assumptions~\ref{ass:omega_boundary_conditions}, \ref{ass:g_bounded}, and~\ref{ass:desired_density} hold. Without loss of generality, we take $\Omega = [-a, a]^n$ with $a>0$. In this higher-dimensional setting, Assumption~\ref{ass:g_bounded} becomes
\begin{align}
    -K_j < g_{i,j}(t, \mathbf{X}(t)) < K_j, \quad \forall\, t\in\mathbb{R}_{\geq 0},
\end{align}
where $g_{i,j}$ denotes the $j$-th component of $\mathbf{g}_i$ and $K_j$ is a positive constant. Similarly, Assumption~\ref{ass:desired_density} requires that $\rho^\mathrm{d}_{x_i},\, \rho^\mathrm{d}_{x_i x_i}\in\mathcal{L}^{\infty}(\Omega)$, where $x_i$ denotes the $i$-th component of the spatial coordinate~$\mathbf{x}$.
\subsection{Control design}

\subsubsection{Bounding systems}
As in the one-dimensional setting, we introduce bounding systems for~\eqref{eq:micro_hd}:
\begin{subequations}\label{eq:bounding_sys}
    \begin{align}
        \mathrm{d}\bar{\mathbf{x}}_i(t) &= \left[\mathbf{u}_i(t) + \mathbf{k}\right] \mathrm{d}t + \sqrt{2D}\, \mathrm{d}\mathbf{W}_i(t), \label{eq:bounding_upper}\\
        \mathrm{d}\underbar{${\mathbf{x}}$}_i(t) &= \left[\mathbf{u}_i(t) - \mathbf{k}\right] \mathrm{d}t + \sqrt{2D}\, \mathrm{d}\mathbf{W}_i(t), \label{eq:bounding_lower}
    \end{align}
\end{subequations}
for $i\in\{1,\dots,N\}$, where $\mathbf{k}$ is the vector of positive constants bounding $\mathbf{g}_i$. For instance, if $n=3$, then $\mathbf{k}=[K_1,K_2,K_3]$. By the same comparison argument used in the one-dimensional setting, any control input $\mathbf{u}_i$ that solves the control problem for the bounding systems also solves the nominal problem.

\subsubsection{Macroscopic model}
As in the one-dimensional setting, we consider a continuum mean-field approximation for the bounding systems~\eqref{eq:bounding_sys}, which we compactly write as 
\begin{align}\label{eq:continuum_model_hd}
    \hat{\rho}_t(\mathbf{x}, t) + \nabla \cdot \left[\hat{\rho}(\mathbf{x}, t) \left(\mathbf{U}(\mathbf{x}, t) + \hat{\mathbf{k}}\right)\right] = D\nabla^2\hat{\rho}(\mathbf{x}, t),
\end{align}
where $\hat{k}_i = \pm K_i$ (our analysis considers all possible sign combinations of $K_i$). The operators $\nabla\cdot(\cdot)$ and $\nabla^2(\cdot)$ denote the divergence and Laplacian, respectively. To comply with Assumption~\ref{ass:omega_boundary_conditions} and ensure zero net mass flux from $\Omega$, we impose boundary conditions of the form
\begin{align}\label{eq:no_flux_bc_hd}
    \begin{cases}
        \displaystyle\int_{\partial\Omega} \left[D\nabla\hat{\rho}(\mathbf{x},t) - \hat{\rho}(\mathbf{x},t)\mathbf{k}\right] \cdot \hat{\mathbf{n}}\,\mathrm{d}\mathbf{x}=0,\\[6pt]
        \displaystyle\int_{\partial\Omega} \left[\hat{\rho}(\mathbf{x},t)\mathbf{U}(\mathbf{x},t)\right] \cdot \hat{\mathbf{n}}\,\mathrm{d}\mathbf{x}=0,
    \end{cases}
\end{align}
where $\hat{\mathbf{n}}$ is the outward unit normal to $\partial\Omega$ and $\nabla(\cdot)$ denotes the gradient operator.
\begin{remark}
    Boundary conditions~\eqref{eq:no_flux_bc_hd} are satisfied if $\hat{\rho}$ and $\mathbf{U}$ are periodic, or under reflective boundary conditions, i.e.,
    \begin{align}\label{eq:reflective_bc_hd}
        \begin{cases}
            \left[D\nabla\hat{\rho}(\mathbf{x},t) - \hat{\rho}(\mathbf{x},t)\mathbf{k}\right] \cdot \hat{\mathbf{n}}=0,\\[6pt]
            \left[\hat{\rho}(\mathbf{x},t)\mathbf{U}(\mathbf{x},t)\right] \cdot \hat{\mathbf{n}}=0,
        \end{cases}
    \end{align}
\end{remark}

\subsubsection{Macroscopic control design}
Let us define the error function $e = \rho^{\mathrm{d}} - \hat{\rho}$. Its dynamics obey
\begin{multline}\label{eq:err_dynamics_hd}
    e_t(\mathbf{x}, t) = q(\mathbf{x}, t) + \mathbf{k} \cdot \nabla\rho^\mathrm{d}(\mathbf{x}) 
    - \mathbf{k} \cdot \nabla e(\mathbf{x}, t) \\
    - D\nabla^2\rho^\mathrm{d}(\mathbf{x}) + D\nabla^2 e(\mathbf{x}, t),
\end{multline}
where we used standard vector identities and defined
\begin{align}
    q(\mathbf{x}, t) = \nabla\cdot\left[\hat{\rho}(\mathbf{x}, t)\, \mathbf{U}(\mathbf{x},t)\right].
\end{align}
Boundary conditions for~\eqref{eq:err_dynamics_hd} follow from~\eqref{eq:no_flux_bc_hd}. Assuming the desired density satisfies compatible boundary conditions, they take the form
\begin{align}\label{eq:err_BC_hd}
    \begin{cases}
        \displaystyle\int_{\partial\Omega}\left[D\nabla e(\mathbf{x}, t) - e(\mathbf{x}, t)\mathbf{k}\right]\cdot\hat{\mathbf{n}} \,\mathrm{d}\mathbf{x} = 0,\\[6pt]
        \displaystyle\int_{\partial\Omega}\left[\hat{\rho}(\mathbf{x}, t)\, \mathbf{U}(\mathbf{x}, t)\right]\cdot\hat{\mathbf{n}}\,\mathrm{d}\mathbf{x} = 0.
    \end{cases}
\end{align}

\begin{theorem}[Global exponential macroscopic convergence]\label{thm:convergence_hd}
    Choose $q$ in~\eqref{eq:err_dynamics_hd} as 
    \begin{align}\label{eq:q_hd}
        q(\mathbf{x}, t) = -k_p\, e(\mathbf{x}, t) - k_s(t)\,\mathrm{sign}[e(\mathbf{x}, t)] + \alpha(t),
    \end{align}
    where $k_p > 0$ is a control gain, $\alpha$ is any bounded function of time, and $k_s$ is a time-varying control gain satisfying
    \begin{align}\label{eq:condition_ks_2D}
        k_s(t) > A + H(t),
    \end{align}
    with
    \begin{subequations}
        \begin{align}
            A &= \sum_{i=1}^{n} \left( K_i \lVert \rho_{x_i}^\mathrm{d}(\cdot) \rVert_\infty + D\, \lVert \rho^\mathrm{d}_{x_i x_i}(\cdot) \rVert_\infty \right),\label{eq:A_hd}\\
            H(t) &= \sum_{i=1}^{n} K_i \lVert e_{x_i}(\cdot, t) \rVert_\infty.\label{eq:H(t)}
        \end{align}
    \end{subequations}
    Then, the solution of~\eqref{eq:err_dynamics_hd} converges globally and exponentially to zero in $\mathcal{L}^2(\Omega)$, with a convergence rate upper bounded by~$k_p$.
\end{theorem}

\begin{proof}
We introduce the candidate Lyapunov functional $V = \tfrac{1}{2}\lVert e \rVert_2^2$, whose time derivative is (dropping explicit functional dependencies for brevity)
\begin{multline}\label{eq:lyap_func_hd}
    V_t = \int_\Omega e\, e_t\,\mathrm{d}\mathbf{x} = \int_\Omega e\, q\,\mathrm{d}\mathbf{x} + \int_\Omega e\,\mathbf{k}\cdot\nabla\rho^\mathrm{d}\,\mathrm{d}\mathbf{x} - \int_\Omega e\,\mathbf{k}\cdot\nabla e\,\mathrm{d}\mathbf{x} \\
    - D\int_\Omega e\,\nabla^2\rho^\mathrm{d}\,\mathrm{d}\mathbf{x} + D\int_\Omega e\,\nabla^2 e\,\mathrm{d}\mathbf{x},
\end{multline}
where we substituted~\eqref{eq:err_dynamics_hd}.

We establish the bounds
\begin{subequations}\label{eq:bounds_hd}
    \begin{align}
        \left\lvert \int_\Omega e\,\mathbf{k}\cdot\nabla\rho^\mathrm{d}\,\mathrm{d}\mathbf{x} \right\rvert &\leq \lVert e \rVert_1 \sum_{i=1}^{n} K_i \lVert \rho_{x_i}^\mathrm{d} \rVert_\infty, \label{eq:bound1}\\
        \left\lvert D\int_\Omega e\,\nabla^2\rho^\mathrm{d}\,\mathrm{d}\mathbf{x} \right\rvert &\leq D\lVert e \rVert_1 \sum_{i=1}^{n} \lVert \rho^\mathrm{d}_{x_i x_i} \rVert_\infty, \label{eq:bound2}
    \end{align}
\end{subequations}
where we applied~\eqref{eq:bound_holder}. Substituting~\eqref{eq:bounds_hd} into~\eqref{eq:lyap_func_hd} yields
\begin{multline}\label{eq:Lyap_bound_hd}
    V_t \leq \int_\Omega e\, q\,\mathrm{d}\mathbf{x} - \int_\Omega e\,\mathbf{k}\cdot\nabla e\,\mathrm{d}\mathbf{x}
    + D\int_\Omega e\,\nabla^2 e\,\mathrm{d}\mathbf{x} \\
    + A \int_\Omega \lvert e \rvert\,\mathrm{d}\mathbf{x},
\end{multline}
where $A$ is defined in~\eqref{eq:A_hd}.

We then establish a sequence of useful bounds. Specifically, we note that
\begin{multline}\label{eq:huge_bound_1}
    D\int_\Omega e\,\nabla^2 e\,\mathrm{d}\mathbf{x} - \int_\Omega e\,\mathbf{k}\cdot\nabla e\,\mathrm{d}\mathbf{x}
    \\\overset{(\mathrm{a})}{=} D\int_\Omega e\,\nabla\cdot\nabla e\,\mathrm{d}\mathbf{x} - \frac{1}{2}\int_\Omega \mathbf{k}\cdot\nabla e^2\,\mathrm{d}\mathbf{x}
    \\\overset{(\mathrm{b})}{=} D\int_{\partial\Omega} e\,\nabla e\cdot\hat{\mathbf{n}}\,\mathrm{d}\mathbf{x} - D\int_\Omega \nabla e\cdot\nabla e\,\mathrm{d}\mathbf{x} - \frac{1}{2}\int_\Omega \nabla\cdot(e^2\mathbf{k})\,\mathrm{d}\mathbf{x},
\end{multline}
where (a) uses the identities $\nabla^2 e = \nabla\cdot\nabla e$ and $-e\,\mathbf{k}\cdot\nabla e = -\tfrac{1}{2}\,\mathbf{k}\cdot\nabla e^2$, while (b) applies the multidimensional integration by parts formula for terms of the form $u\,\nabla\cdot\mathbf{F}$ and the identity $\mathbf{k}\cdot\nabla e^2 = \nabla\cdot(\mathbf{k}\, e^2)$ (valid since $\mathbf{k}$ is constant). For the right-hand side of~\eqref{eq:huge_bound_1}, we have
\begin{multline}\label{eq:huge_bound_2}
    D\int_{\partial\Omega} e\,\nabla e\cdot\hat{\mathbf{n}}\,\mathrm{d}\mathbf{x} - D\int_\Omega \nabla e\cdot\nabla e\,\mathrm{d}\mathbf{x} - \frac{1}{2}\int_\Omega \nabla\cdot(e^2\mathbf{k})\,\mathrm{d}\mathbf{x} \\
    \overset{(\mathrm{c})}{\leq} D\int_{\partial\Omega} e\,\nabla e\cdot\hat{\mathbf{n}}\,\mathrm{d}\mathbf{x} - \frac{1}{2}\int_\Omega \nabla\cdot(e^2\mathbf{k})\,\mathrm{d}\mathbf{x}
    \\\overset{(\mathrm{d})}{=} D\int_{\partial\Omega} e\,\nabla e\cdot\hat{\mathbf{n}}\,\mathrm{d}\mathbf{x} - \frac{1}{2}\int_{\partial\Omega} e^2\mathbf{k}\cdot\hat{\mathbf{n}}\,\mathrm{d}\mathbf{x} - \frac{1}{2}\int_{\partial\Omega} e^2\mathbf{k}\cdot\hat{\mathbf{n}}\,\mathrm{d}\mathbf{x}
    \\+ \frac{1}{2}\int_{\partial\Omega} e^2\mathbf{k}\cdot\hat{\mathbf{n}}\,\mathrm{d}\mathbf{x}
    \\\overset{(\mathrm{e})}{=} \int_{\partial\Omega} e\left(D\nabla e - e\,\mathbf{k}\right)\cdot\hat{\mathbf{n}}\,\mathrm{d}\mathbf{x} + \frac{1}{2}\int_{\partial\Omega} e^2\mathbf{k}\cdot\hat{\mathbf{n}}\,\mathrm{d}\mathbf{x}
    \\\overset{(\mathrm{f})}{=} \frac{1}{2}\int_{\partial\Omega} e^2\mathbf{k}\cdot\hat{\mathbf{n}}\,\mathrm{d}\mathbf{x},
\end{multline}

\begin{figure*}[!t]
    \centering
    \includegraphics[width=1\linewidth]{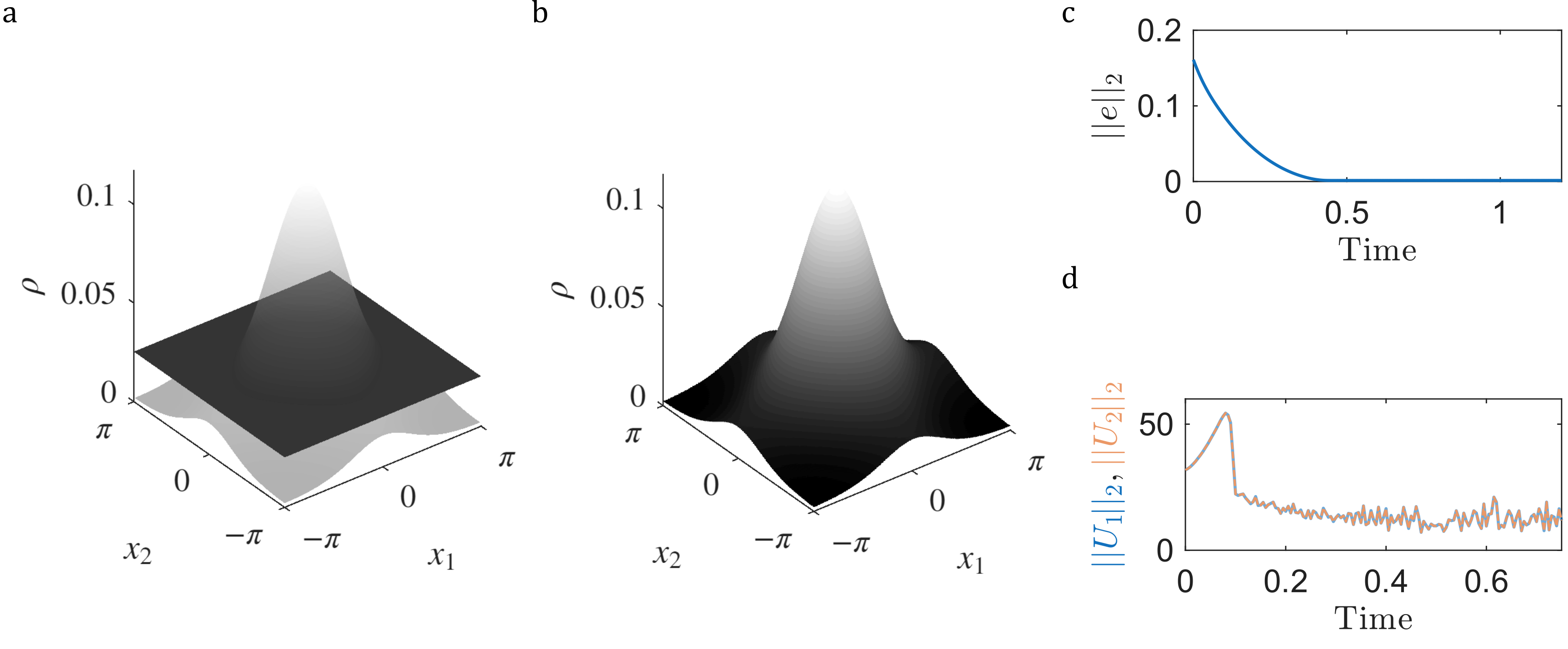}
    \caption{\textbf{Control of a 2-dimensional Fokker--Planck equation.} \textbf{a.}~Plot of the desired density (shaded surface) and initial density $\rho(\mathbf{x},0)$ (solid surface). Darker colors correspond to lower density values, while lighter colors indicate higher densities.
    \textbf{b.}~Plot of the desired density (shaded surface) and final density $\rho(\mathbf{x},T_f)$ (solid surface), where $T_f$ denotes the final simulation time.
    \textbf{c.}~Evolution of $\lVert e \rVert_2$ over time.
    \textbf{d.}~Evolution in time of the $\mathcal{L}^2$ norm of the first (blue solid line) and second (orange dashed line) component of the velocity field $\mathbf{U}(\mathbf{x},t)$.}
    \label{fig:2D_Validation_Continuum}
\end{figure*}
where
\begin{itemize}
    \item in (c), we exploited that $-D\int_\Omega \nabla e\cdot\nabla e\,\mathrm{d}\mathbf{x} \leq 0$;
    \item in (d), we applied the divergence theorem to the second term and added and subtracted $\tfrac{1}{2}\int_{\partial\Omega} e^2\mathbf{k}\cdot\hat{\mathbf{n}}\,\mathrm{d}\mathbf{x}$;
    \item in (e), we rearranged the terms;
    \item in (f), we applied boundary conditions~\eqref{eq:err_BC_hd}.
\end{itemize}
Combining~\eqref{eq:huge_bound_1} and~\eqref{eq:huge_bound_2} yields
\begin{align}\label{eq:huge_bound}
    D\int_\Omega e\,\nabla^2 e\,\mathrm{d}\mathbf{x} - \int_\Omega e\,\mathbf{k}\cdot\nabla e\,\mathrm{d}\mathbf{x} \leq \frac{1}{2}\int_{\partial\Omega} e^2\mathbf{k}\cdot\hat{\mathbf{n}}\,\mathrm{d}\mathbf{x}.
\end{align}
For the right-hand side of~\eqref{eq:huge_bound}, we observe that
\begin{multline}\label{eq:small_bound_hd}
    \left\lvert \frac{1}{2}\int_{\partial\Omega} e^2\mathbf{k}\cdot\hat{\mathbf{n}}\,\mathrm{d}\mathbf{x} \right\rvert = \frac{1}{2}\left\lvert \int_\Omega \nabla\cdot(e^2\mathbf{k})\,\mathrm{d}\mathbf{x} \right\rvert
    \\= \sum_{i=1}^{n} \left\lvert \int_\Omega K_i\, e\, e_{x_i}\,\mathrm{d}\mathbf{x} \right\rvert \leq \sum_{i=1}^{n} \lVert K_i\, e\, e_{x_i} \rVert_1
    \\\leq \lVert e \rVert_1 \sum_{i=1}^{n} K_i \lVert e_{x_i} \rVert_\infty,
\end{multline}
where we applied the divergence theorem, the definition of the $\mathcal{L}^1(\Omega)$ norm, and~\eqref{eq:bound_holder}.

Substituting~\eqref{eq:small_bound_hd} into~\eqref{eq:huge_bound}, and then into~\eqref{eq:Lyap_bound_hd}, yields
\begin{align}\label{eq:bound_on_Vt}
    V_t \leq \int_\Omega e\, q\,\mathrm{d}\mathbf{x} + \left[A + H\right] \int_\Omega \lvert e \rvert\,\mathrm{d}\mathbf{x},
\end{align}
where $H$ is defined in~\eqref{eq:H(t)}.

Hence, choosing $q$ as in~\eqref{eq:q_hd} ensures global exponential stability in $\mathcal{L}^2(\Omega)$---for more details, see the one-dimensional setting in Theorem~\ref{th:convergence_1d}.
\end{proof}
\begin{remark}
    Similarly to the one-dimensional setting, considering a periodic domain with periodic boundary conditions on $\hat{\rho}$ and $U$ simplify derivations. Specifically, it is easy to verify that
    \begin{align}
        \frac{1}{2}\int_{\partial\Omega}e^2(\mathbf{x},t)\mathbf{k}\cdot\hat{\mathbf{n}} \,\mathrm{d}\mathbf{x}=0,
    \end{align}
    ensuring that \eqref{eq:huge_bound} simplifies to
    \begin{align}
        D\int_\Omega e(\mathbf{x},t)\nabla^2e(\mathbf{x},t)\,\mathrm{d}\mathbf{x} - \int_\Omega e(\mathbf{x},t)\mathbf{k} \cdot \nabla e(\mathbf{x},t)\,\mathrm{d}\mathbf{x} \leq 0.
    \end{align}
    Consequently, \eqref{eq:bound_on_Vt} simplifies to
    \begin{align}
        V_t(t) \leq \int_\Omega e(\mathbf{x},t) q(\mathbf{x},t) \,\mathrm{d}\mathbf{x} + A \int_\Omega\vert e(\mathbf{x},t) \vert\,\mathrm{d}\mathbf{x},
    \end{align}
    requiring only
    \begin{align}
        k_s(t) = k_s>A
    \end{align}
    in \eqref{eq:q_hd} to ensure global asymptotic convergence.
\end{remark}

\begin{remark}
    Similarly to the one-dimensional setting, $\alpha$ needs to be chosen to enforce the control action fulfills boundary conditions. Details on how to choose $\alpha$ are derived in the next section. 
\end{remark}

\subsubsection{Discretization} \label{sec:discretization_2d}
The relation $q=\nabla \cdot \mathbf{\Phi}$, where $\mathbf{\Phi} = \hat{\rho}\mathbf{U}$, does not suffice to uniquely determine $\mathbf{U}$ from $q$, being $q$ a scalar function and $\mathbf{U}$ a vector field. 
For this reason, as done in \cite{maffettone2024mixed} for a related problem, we impose a zero curl condition to close the problem\footnote{Such a choice is arbitrary and other closures to the problem can be used.}. This results in 
\begin{subequations}\label{eq:div_curl}
    \begin{align}
        &\nabla \cdot\mathbf{\Phi}(\mathbf{x},t) = q(\mathbf{x},t )\\
        &\nabla \times \mathbf{\Phi}(\mathbf{x},t) =0
    \end{align}
\end{subequations}
with $\int_{\partial\Omega}\Phi \cdot \hat{\mathbf{n}} = 0$ on $\partial\Omega$ (to cope with boundary conditions \eqref{eq:no_flux_bc_hd}). Introducing the scalar potential $\xi$, such that $\mathbf{\Phi}=-\nabla \xi$, and being $\Omega$ simply connected, \eqref{eq:div_curl} can be recast as the Poisson equation
\begin{align}\label{eq:poisson}
    \nabla^2\xi(\mathbf{x}, t) = -q(\mathbf{x}, t) 
\end{align}
with $\int_{\partial\Omega}\nabla\xi \cdot \hat{\mathbf{n}} = 0$ on $\partial\Omega$. For the problem to admit a solution, we need to fulfill the compatibility condition $\int_\Omega q\,\mathrm{d}\mathbf{x} =0$ \cite{evans2022partial}. In particular, integrating the left-hand side of \eqref{eq:poisson}, applying the divergence theorem and boundary conditions, we find
\begin{align}
    \int_\Omega \nabla^2\xi\,\mathrm{d}\mathbf{x} = \int_\Omega \nabla\cdot\nabla\xi\,\mathrm{d}\mathbf{x} = \int_{\partial\Omega} \nabla\xi\cdot\hat{\mathbf{n}} \,\mathrm{d}\mathbf{x} = 0.
\end{align}
For the right-hand side of \eqref{eq:poisson} to integrate to 0 as well, it needs to be $\int_\Omega q\,\mathrm{d}\mathbf{x} = 0$.

We can choose $\alpha$ in \eqref{eq:q_hd}  to enforce such a compatibility constraint, that is 
    \begin{align}
        \alpha(t) = \frac{\int_\Omega \mathrm{sign}[e(\mathbf{x}, t)]\,\mathrm{d}\mathbf{x}}{\vert \Omega \vert}.
    \end{align}

Problem \eqref{eq:poisson} can be solved by cosine Fourier series expansion\footnote{For periodic domains and periodic functions, classic Fourier series expansion can be used.}. Specifically, to cope with the boundary condition, we expand $\xi$ into its cosine Fourier series, thus enforcing that $\int_{\partial\Omega}\nabla\xi\cdot\hat{\mathbf{n}} = 0$. This results in (we fix $n=2$ without any loss of generality)
\begin{align}\label{eq:xi_exp}
    \xi(x_1, x_2) = \sum_{h,k = 1}^\infty a_{hk}\cos\left(\frac{h\pi (x_1+a)}{2a}\right) \cos\left(\frac{k\pi (x_2+a)}{2a}\right),
\end{align}
so that the Laplacian can be written as
\begin{align}\label{eq:laplacian_exp}
    &\nonumber\nabla^2\xi(x_1, x_2) = -\sum_{h,k = 1}^\infty a_{hk} \left[\left(\frac{h\pi}{2a}\right)^2+\left(\frac{k\pi}{2a}\right)^2\right]\\
    &\cos\left(\frac{h\pi (x_1+a)}{2a}\right) \cos\left(\frac{k\pi (x_2+a)}{2a}\right)
\end{align}
Note that we are assuming $\xi$ to be evenly replicated outside of $\Omega$, so as to not require any symmetry within $\Omega$. The basis $\cos\left(\frac{h\pi(x_1+a)}{2a}\right), \cos\left(\frac{k\pi(x_2+a)}{2a}\right)$ is complete and orthonormal for functions that are defined in $\Omega$ and evenly replicated outside of $\Omega$.

Similarly we can expand $q$ into its cosine Fourier series
resulting in
\begin{align}\label{eq:q_exp}
    q(x_1, x_2) = \sum_{h,k = 1}^\infty \alpha_{hk}\cos\left(\frac{h\pi (x_1+a)}{2a}\right) \cos\left(\frac{k\pi (x_2+a)}{2a}\right),
\end{align}
where $\alpha_{nm}$ are known Fourier coefficients\footnote{Similarly with respect to $\xi$, we assume $q$ to be defined in $\Omega$ and evenly replicated in each direction outside of $\Omega$}.
Finally by equating \eqref{eq:laplacian_exp} and \eqref{eq:q_exp} we find that $\xi$ is given by \eqref{eq:xi_exp} and the coefficients can be computed starting from those of $q$ as
\begin{align}
    a_{hk} = \frac{\alpha_{hk}}{\left(\frac{h\pi}{2a}\right)^2+\left(\frac{k\pi}{2a}\right)^2}.
\end{align}
For implementation purposes infinite series are to be truncated.

\begin{remark}
    Solving \eqref{eq:poisson} using Fourier Series (classic Fourier series for periodic boundary conditions and cosine Fourier series for reflective boundaries) inherently ensures the resulting control flux $\mathbf{\Phi}$ fulfills boundary conditions.
\end{remark}

{ Once $\mathbf{\Phi}$ is available, the macroscopic velocity field can be recovered as $\mathbf{U} = \mathbf{\Phi}/\hat{\rho}$. The microscopic control inputs for~\eqref{eq:micro_hd} are then obtained by spatial sampling as
\begin{align}\label{eq:spatial_sampling_hd}
    \mathbf{u}_i(t) = \mathbf{U}(\mathbf{x}_i(t), t), \quad i \in \{1, \dots, N\}.
\end{align}
}

\subsection{Numerical validation}
As in the one-dimensional case, we first validate our framework in the continuum macroscopic setting by integrating~\eqref{eq:continuum_model_hd}, corresponding to the regime $N\to\infty$. {Subsequently, in Section~\ref{subsec:swarm_robotics},} we validate the strategy for a finite-size multi-agent system {through a representative swarm robotics application}.

We consider~\eqref{eq:continuum_model_hd} in a two-dimensional periodic domain $\Omega=[-\pi,\pi]^2$, setting $D = 0.1$ and $K_1=K_2=1$. As reference density, we choose the bivariate von~Mises distribution
\begin{multline}
    \rho^\mathrm{d}(x_1, x_2) = \exp\!\big(\kappa_1\cos(x_1-\mu)+\kappa_2\cos(x_2-\nu)\\
    +\cos(x_1-\mu)\cos(x_2-\nu)+\sin(x_1-\mu)\sin(x_2-\nu)\big),
\end{multline}
where $\mu = 0$ and $\nu = 0$ are the means along $x_1$ and $x_2$, respectively, and $\kappa_1 = 1$ and $\kappa_2 = 1$ are concentration parameters. The control gains are $k_p=1$ and $k_s = 1.5\, A$, with $A$ defined in~\eqref{eq:A_hd}.

Equation~\eqref{eq:continuum_model_hd} is integrated using an operator-splitting technique: the advection part is discretized via a Lax--Friedrichs scheme, while diffusion is handled with a Crank--Nicolson method~\cite{leveque2002finite}. The spatial domain is discretized on a $200\times200$ uniform grid with a time step of $\Delta t = 5\cdot10^{-4}$.

Starting from a uniform distribution (see Fig.~\ref{fig:2D_Validation_Continuum}a), the controller~{\eqref{eq:q_hd}} steers the density of the bounding system towards $\rho^{\mathrm{d}}$, as shown by the final density profile in Fig.~\ref{fig:2D_Validation_Continuum}b. We quantitatively assess performance by analyzing the $\mathcal{L}^2(\Omega)$ norms of the control error and of the control velocity field. As expected, the error norm decreases monotonically, approaching zero in approximately $0.4$ time units (see Fig.~\ref{fig:2D_Validation_Continuum}c). As in the one-dimensional setting, the velocity field norm does not vanish when the error reaches zero, owing to the persistent need to compensate for diffusion and the constant drift (see Fig.~\ref{fig:2D_Validation_Continuum}d).

\section{Applications}
\label{sec:applications}

Having established theoretical convergence guarantees and validated them numerically, 
we demonstrate the versatility of the proposed framework through two relevant applications: 
traffic flow regulation and swarm robotics in partially unknown environments. 
These examples illustrate that our control design naturally extends to systems 
with realistic nonlinear dynamics and environmental disturbances, while maintaining 
the macroscopic convergence properties proven in Section~\ref{sec:control_design} and Section~\ref{sec:2D_derivation_and_validation}.

\subsection{Traffic Flow Regulation}
\label{subsec:traffic}
We consider an ensemble of vehicles moving on a ring according to an optimal velocity car-following model~\cite{lazar2016review}. This application demonstrates the framework's ability to regulate traffic systems with realistic microscopic dynamics and time-varying macroscopic objectives, extending beyond the stationary desired densities considered in the previous validation examples. {As in~\cite{fueyo2025continuation}, we consider the case in which all vehicles in the platoon can be controlled.}

Denoting by $p_i$ and $v_i$ the position and velocity of the $i$-th vehicle, respectively, we describe its dynamics as
\begin{subequations}
    \begin{align}
        \dot{p}_i &= v_i,\\
        \tau\, \dot{v}_i &= \left(v_{\mathrm{opt}}(s_i) + u_i\right) - v_i,
    \end{align}
\end{subequations}
where $\tau$ is a relaxation factor, $v_{\mathrm{opt}}(s_i) = v_{\max}\, \frac{\tanh(s_i/\Delta s - \beta) + \tanh\beta}{1 + \tanh\beta}$ is the optimal velocity, $s_i = p_{i+1} - p_i$ is the spacing between successive vehicles, $\Delta s$ is the transition width, and $\beta$ is the form factor. Due to the ring topology, periodic boundary conditions are enforced on the vehicle positions, with $p_i \in [-\pi, \pi]$ for simplicity. Additionally, the speed of each vehicle is bounded, i.e., $v_i \in [0, v_{\max}]$, where $v_{\max}$ denotes the maximum allowed velocity, resulting in reflective boundary conditions for the velocity.

In the limit $\tau \to 0$, the vehicle dynamics reduce to
\begin{equation}
    \dot{p}_i = v_{\mathrm{opt}}(s_i) + u_i,
\end{equation}
subject to the constraint $v_{\mathrm{opt}}(s_i) + u_i \in [0, v_{\max}]$.

We recast the vehicle dynamics in the form of~\eqref{eq:micro_system} by letting $x_i = p_i$ and setting $g_i(X(t), t) = v_{\mathrm{opt}}(s_i(x_i, x_{i+1}))$. Here, $g_i$ captures the effect of surrounding vehicles on vehicle~$i$. Since $v_{\mathrm{opt}} \in [0, v_{\max}]$ and $x_i \in [-\pi,\pi]$, boundedness of $g_i$ follows directly, and choosing $K = \, v_{\max}$ provides a uniform bound on the drift of each vehicle. Note that no stochastic term is present in this scenario.

We simulate a network of $N = 20\,000$ vehicles on the ring using a second-order Runge--Kutta method with time step $\Delta t = 0.01$, maximum velocity $v_{\max} = 10$, relaxation factor $\tau = 0.01$, and optimal headway $\beta = 0.5$. The control gains are set to $k_p = 10$ and $k_s(t) = 1.5\, A + K\lVert e_x \rVert_\infty$, with $K = v_{\max}$, which satisfies the hypotheses of Theorem~\ref{th:convergence_1d}. As desired velocity distribution, we choose~\eqref{eq:Von_Mieses_Distribution} with time-varying mean $\mu(t) = 0.5\, t$ and concentration $\kappa = 0.5$. To accommodate the time-varying reference $\rho^\mathrm{d}$, we incorporate a feedforward term accounting for $\rho_t^\mathrm{d}$ in $q$, as described in Section~\ref{sec:Tracking_problem}.
\begin{figure*}[!t]
    \centering
    \includegraphics[width=1\linewidth]{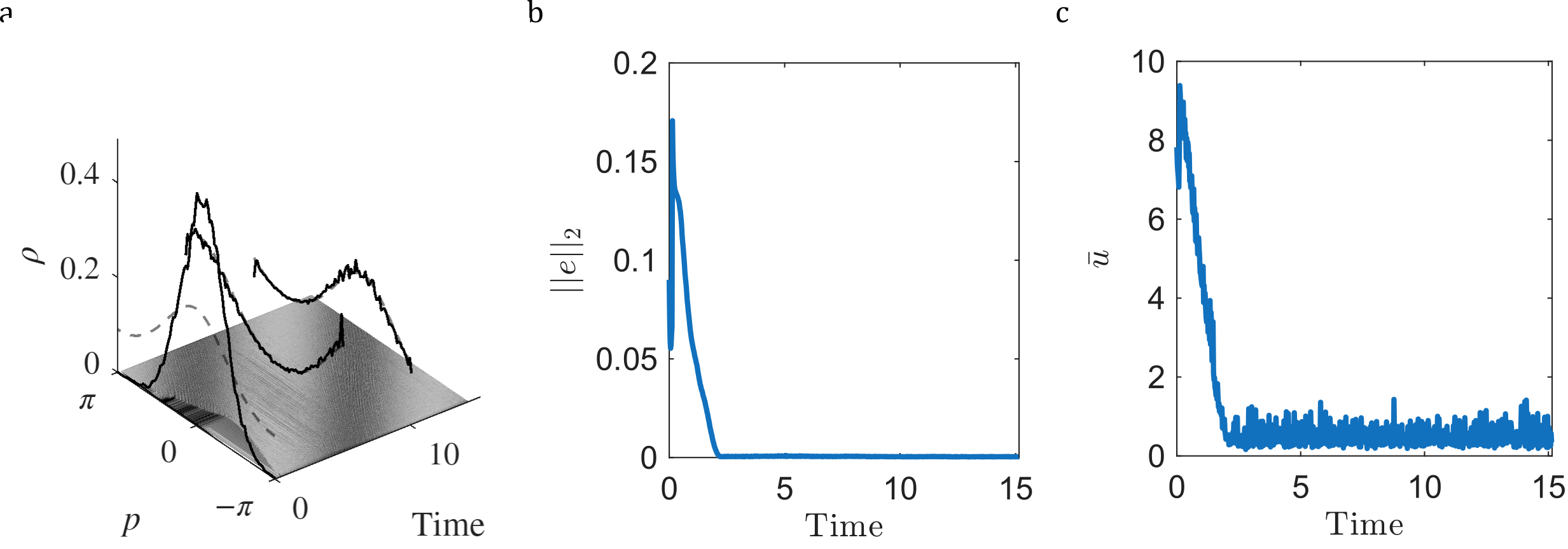}
    \caption{\textbf{Control of vehicles moving on a ring}.  \textbf{a}. Evolution in time an space of all the agents in the ensemble (x and y axes). On the z axis the estimated (solid) and desired (dashed) densities are displayed in four representative time instants. 
    \textbf{b}. Evolution of the $\mathcal{L}^2$ norm of the control error in time.
    \textbf{c}. Evolution of the average velocity of the vehicles over time. This metric was computed as $\bar u(t) = 1/N \sum_{i=1}^{N}{|u_i(t)|}$, with $N$ being the number of vehicles in the group.  
    }
    \label{fig:1D_Traffic}
\end{figure*}

Given the finite number of agents, {the macroscopic control input~\eqref{eq:q} is computed by estimating} the density through normalization and Gaussian kernel filtering of the histogram of agents' velocities on the spatial grid. {The resulting $q$ is then integrated in space to recover $U$, which is discretized via~\eqref{eq:spatial_sampling}.}

As shown in Fig.~\ref{fig:1D_Traffic}a, the vehicle distribution is rapidly steered towards the desired profile, achieving stabilization despite the presence of nonlinear unmodeled dynamics. The $\mathcal{L}^2(\Omega)$ norm of the control error, reported in Fig.~\ref{fig:1D_Traffic}b, reaches a steady-state value of $4.24 \times 10^{-4}$ in approximately $2.2$ time units. This residual mismatch is attributed to finite-size effects arising from the limited number of agents. Notice that, due to the presence of a strong saturation in the control input, during the transient there is a short time interval where the $\mathcal{L}^2(\Omega)$ of the control error increases. 

We complement the analysis by examining the control effort, quantified by the time evolution of the average control input $\bar{u}(t) = \frac{1}{N}\sum_{i=1}^{N} |u_i(t)|$. As shown in Fig.~\ref{fig:1D_Traffic}c, the average control input decreases as the density approaches $\rho^\mathrm{d}$ but does not vanish at steady state.

\begin{figure*}[!t]
    \centering
    \includegraphics[width=1\linewidth]{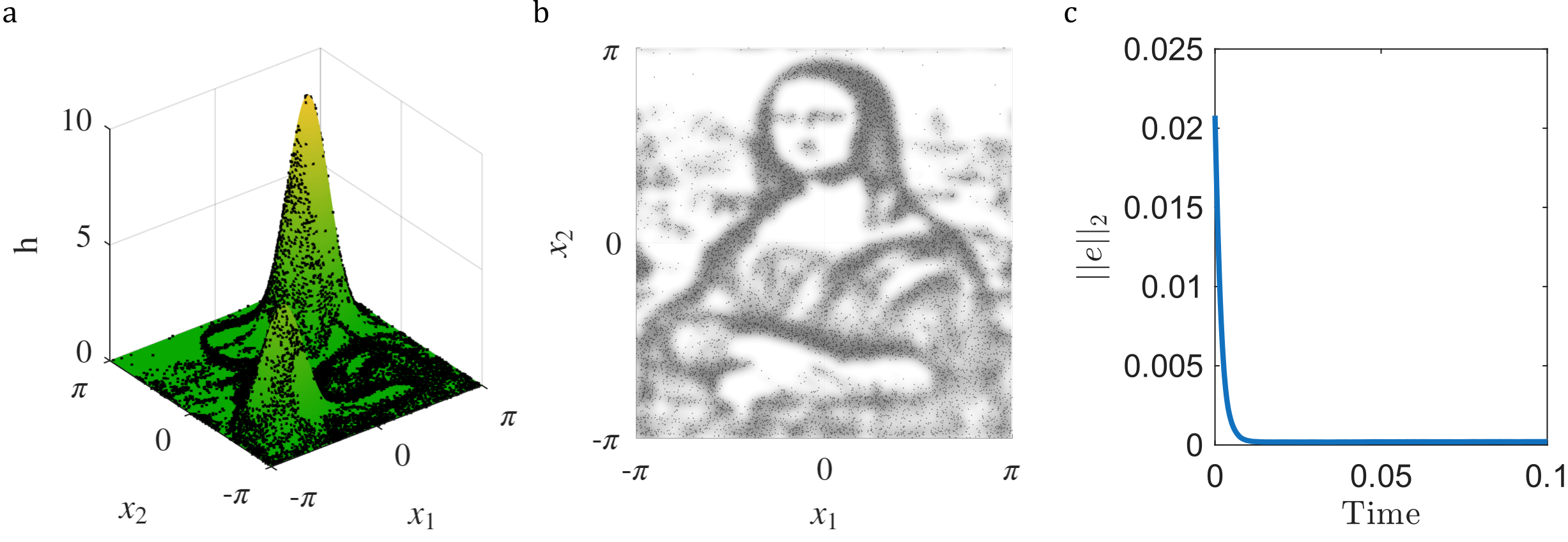}
    \caption{\textbf{Control of an ensemble of Unmanned Ground Vehicles} \textbf{a.} position of the unmanned ground vehicles (represented as black dots) in the hilly landscape. The elevation of the ground is graphically represented as a surface with a color map where the color shifts from green to orange as the ground elevation increases.
    \textbf{b.} Agents positions at the final simulation step against the reference density. The reference density is portrayed as a greyscale shaded surface where the higher is the density the darker is the color.
    \textbf{c.} Evolution of $\mathcal{L}^2$ norm of the control error in time.}
    \label{fig:2D_Validation_UGV}
\end{figure*}

\subsection{Swarm Robotics in Partially Unknown Environments}
\label{subsec:swarm_robotics}
We consider a swarm of unmanned ground vehicles (UGVs) tasked with achieving a prescribed spatial formation in an environment with unknown terrain elevation. This scenario is representative of applications such as foraging (search and rescue) \cite{ornia2022mean}, environmental monitoring \cite{kakalis2008robotic}, and collaborative construction \cite{petersen2019review},where the terrain profile is not fully known in advance.

Assuming negligible robot inertia and an inner control loop compensating for nonholonomic constraints and friction, the dynamics of the $i$-th robot are given by
\begin{subequations}
    \begin{align}
        \dot{x}_{1,i} &= u_{1,i} + g_1(x_{1,i}, x_{2,i}), \\
        \dot{x}_{2,i} &= u_{2,i} + g_2(x_{1,i}, x_{2,i}),
    \end{align}
\end{subequations}
where $g_1$ and $g_2$ model the effect of the uneven terrain on the speed of each robot.

We consider a periodic domain $\Omega = [-\pi, \pi]^2$ representing a hilly landscape, in which the terrain-induced disturbance is captured by the gradient of a potential whose maxima coincide with the hilltops. Specifically,
\begin{multline}\label{eq:dist_potential}
    \mathbf{G}(x_1, x_2) = [g_1(x_1, x_2),\, g_2(x_1, x_2)]^\top \\
    = \nabla\Big(h_1\, e^{-2[(x_1 - \frac{\pi}{2})^2 + (x_2 - \frac{\pi}{2})^2]}\\
    + h_2\, e^{-2[(x_1 + \frac{\pi}{2})^2 + (x_2 + \frac{\pi}{2})^2]}\Big),
\end{multline}
representing two hills of height $h_1$ and $h_2$ centered at $(\pi/2,\, \pi/2)$ and $(-\pi/2,\, -\pi/2)$, respectively. We set $h_1 = 5$ and $h_2 = 10$. Choosing $K_1 = h_1$ and $K_2 = h_2$ provides a uniform bound on the disturbance acting on the system.

In this application, we aim to steer the swarm toward a prescribed static spatial pattern. To this end, the desired density $\rho^\mathrm{d}$ is constructed as a smoothed and normalized version of a grayscale image of the Mona Lisa, ensuring nonnegativity and unit mass. The control gains are set to $k_p = 1$ and $k_s = 1.1A$, where $A$ is defined in~\eqref{eq:A_hd}. The Mona Lisa image is selected as a representative complex spatial pattern characterized by complex features and irregular geometry. This choice serves to illustrate the ability of the proposed framework to regulate highly nonconvex and non-separable density profiles under structured environmental disturbances.

{To compute the control input $q$ in~\eqref{eq:q_hd},} we estimate the density by normalizing and applying a two-dimensional Gaussian kernel filter to the histogram of agents' positions on the spatial grid. {The velocity field $\mathbf{U}$ is then recovered by spatially integrating $q$ and discretized via~\eqref{eq:spatial_sampling_hd}.} The domain is discretized on a $200 \times 200$ uniform grid, and the dynamics of $N = 10\,000$ agents are integrated using a forward Euler scheme with time step $\Delta t = 10^{-4}$. Initial positions are drawn from a uniform distribution.

As shown in Fig.~\ref{fig:2D_Validation_UGV}a, at the final simulation instant the robots arrange on the spatially uneven domain to reproduce the Mona Lisa when viewed from above. For a more detailed comparison, Fig.~\ref{fig:2D_Validation_UGV}b overlays the agent positions at steady state with the reference density, portrayed as a shaded grayscale surface, confirming close qualitative agreement between $\rho^\mathrm{d}$ and the spatial distribution of the agents. Finally, the $\mathcal{L}^2(\Omega)$ norm of the control error, shown in Fig.~\ref{fig:2D_Validation_UGV}c, decreases monotonically and settles in approximately $0.01$ time units.

These applications demonstrate the broad applicability of the proposed control framework. The ability to accommodate realistic microscopic dynamics, time-varying objectives, and environmental uncertainties---while maintaining rigorous convergence guarantees---makes the approach well suited for practical deployment in traffic management, robotic coordination, and related domains.

\section{Conclusions}\label{sec:Discussion}

In this work, we proposed a robust macroscopic control framework for large-scale multi-agent systems affected by unmodeled dynamics and disturbances. Within a macro-to-micro control architecture, we designed a feedback law at the continuum level that steers the population density towards any desired profile compatible with the boundary conditions of the domain. The methodology guarantees global exponential convergence of the density tracking error in $\mathcal{L}^2(\Omega)$ for a suitable choice of the control gains. In particular, provided that perturbations are bounded, the proposed control law ensures asymptotic rejection of unknown drift terms at the macroscopic level, thereby establishing population-level robustness margins prior to microscopic implementation.

The controller has been validated numerically, both at the macroscopic (continuum) and microscopic (agent-based) levels, across a range of application-oriented scenarios. The generality and versatility of the framework are confirmed by experiments spanning density tracking of heterogeneous oscillators, density profile control in traffic systems, and formation of complex spatial patterns with swarms of unmanned ground vehicles operating on terrain with unknown elevation. These results indicate that the macroscopic robustness guarantees are preserved under spatial sampling and finite-size effects, reinforcing the viability of the proposed multi-scale control strategy.

Despite these promising results, the proposed approach presents some limitations that motivate further research. First, although our numerical experiments provide strong evidence that finite-size populations can be steered towards the desired density, the convergence proof relies on the mean-field assumption $N \to \infty$. Future work will aim to analytically characterize how macroscopic convergence guarantees translate to the agent-based level within a unified multi-scale analysis. Second, the controller is centralized, requiring knowledge of the density at every point in the domain. Developing decentralized implementations that rely only on local sensing capabilities is an important direction for scalable deployment. Third, the current formulation assumes that the control input acts on all state variables of each agent, which may not be feasible in every application. Investigating controllability conditions and control design under actuation constraints would significantly broaden the class of systems amenable to this approach.

Notwithstanding these limitations, the proposed architecture contributes to the rigorous foundations of multi-scale control of large, uncertain, and stochastic multi-agent systems. By combining continuum design, robustness margins, and systematic micro-level implementation, this framework provides a scalable methodology that can be integrated with safety constraints, learning-based components, and experimental platforms~\cite{giusti2025data, maffettone2024mixed}. A successful experimental demonstration would enable robust coordination of large-scale swarms, with applications in robotics~\cite{dorigo2021swarm}, environmental management~\cite{zahugi2013oil}, and biotechnology~\cite{massana2022rectification}.

\section*{References}
\bibliographystyle{ieeetr}

\end{document}